\documentclass[12pt,a4paper]{article}
\usepackage[T1]{fontenc}
\usepackage[utf8]{inputenc}
\usepackage{mathpazo}
\usepackage[scaled=0.92]{helvet}
\usepackage{microtype}
\usepackage[a4paper,top=1.25in,bottom=1.10in,left=1.15in,right=1.15in,headheight=15pt]{geometry}
\usepackage{setspace}
\usepackage{titlesec}
\titleformat{\section}{\large\bfseries\singlespacing}{\thesection}{1em}{}
\titleformat{\subsection}{\normalsize\bfseries\singlespacing}{\thesubsection}{1em}{}
\titleformat{\subsubsection}{\normalsize\itshape\singlespacing}{\thesubsubsection}{1em}{}
\titleformat{\paragraph}[runin]{\normalsize\bfseries}{}{0pt}{}[.\enspace]
\titlespacing*{\section}     {0pt}{20pt plus 3pt minus 2pt}{8pt plus 1pt}
\titlespacing*{\subsection}  {0pt}{14pt plus 2pt minus 2pt}{4pt plus 1pt}
\titlespacing*{\subsubsection}{0pt}{10pt plus 2pt minus 1pt}{3pt plus 1pt}
\titlespacing*{\paragraph}   {0pt}{7pt  plus 1pt minus 1pt}{0pt}
\usepackage{amsmath,amssymb,amsthm}
\usepackage{mathtools}
\usepackage{bm}
\newtheorem{proposition}{Proposition}
\newtheorem{corollary}{Corollary}[proposition]
\newtheorem{definition}{Definition}
\usepackage{booktabs}
\usepackage{multirow}
\usepackage{array}
\usepackage{adjustbox}
\usepackage{colortbl}

\usepackage{etoolbox}
\AtBeginEnvironment{table}{\singlespacing\small}
\AtBeginEnvironment{figure}{\singlespacing}
\usepackage{caption}
\usepackage{subcaption}
\newcolumntype{L}[1]{>{\raggedright\arraybackslash}p{#1}}
\newcolumntype{C}[1]{>{\centering\arraybackslash}p{#1}}
\newcolumntype{R}[1]{>{\raggedleft\arraybackslash}p{#1}}
\usepackage{graphicx}
\usepackage{float}
\usepackage{pdflscape}
\usepackage[dvipsnames,svgnames,table]{xcolor}
\definecolor{citeblue}{RGB}{31,119,180}
\definecolor{linkgray}{RGB}{80,80,80}
\usepackage[hidelinks,colorlinks=true,linkcolor=linkgray,citecolor=citeblue,urlcolor=citeblue]{hyperref}
\usepackage[authoryear,round,sort]{natbib}
\usepackage{enumitem}
\setlist[itemize]{itemsep=2pt,topsep=4pt,parsep=0pt,leftmargin=1.6em}
\setlist[enumerate]{itemsep=2pt,topsep=4pt,parsep=0pt,leftmargin=1.6em}
\usepackage{footmisc}
\renewcommand{\footnoterule}{\kern-3pt\hrule width 0.38\textwidth \kern 2.6pt}
\usepackage{appendix}
\usepackage{booktabs}
\usepackage{tabularx}
\usepackage{threeparttable}
\usepackage{multirow}
\usepackage{array}

\usepackage{graphicx}
\usepackage{pgfplots}
\pgfplotsset{compat=1.18}
\usepackage{pgfplotstable}
\usepackage[labelfont=bf, font=small, skip=4pt]{caption}
\usepackage{subcaption}

\usepackage{xcolor}
\definecolor{myblue}{RGB}{31,119,180}
\definecolor{myred}{RGB}{214,39,40}
\definecolor{mygreen}{RGB}{44,160,44}
\definecolor{myorange}{RGB}{255,127,14}
\definecolor{mypurple}{RGB}{148,103,189}
\definecolor{lightblue}{RGB}{174,214,241}
\definecolor{rowgray}{gray}{0.93}

\title{\textbf{Global Reallocation of Capital in the Era of Geo-Economic Fragmentation and Artificial Intelligence}}
\date{\small May 2026\\[4pt]
}
\author{Fernando Toledo\footnote{
Universidad Nacional de La Plata, La Plata, Argentina and Universidad Argentina de la Empresa (UADE), Instituto de Economia (INECO), Ciudad Autónoma de Buenos Aires, Argentina.  fernando.toledo@econo.unlp.edu.ar}  \and Gabriel Montes-Rojas\footnote{
Instituto Interdisciplinario de Econom\'ia Pol\'itica, Universidad de Buenos Aires-CONICET, 
and Universidad Argentina de la Empresa (UADE), Instituto de Economia (INECO), Ciudad Autónoma de Buenos Aires, Argentina. 
gabriel.montes@economicas.uba.ar}  }

\begin{document}
\maketitle
\thispagestyle{plain}

\begin{abstract}
\noindent
This paper studies how geo-economic fragmentation affects the international allocation of capital in an economy where artificial intelligence (AI) relies on specialized, mobile investment. We develop a heterogeneous-agent open-economy model with AI-specific capital and geopolitical frictions in cross-border asset holdings. These frictions reduce capital mobility, raise the cost of AI adoption, and generate distributional effects across countries and households. The quantitative results show that fragmentation produces welfare losses that rise nonlinearly with geopolitical distance. These losses are larger for AI-lagging economies and for low-wealth households, mainly through reduced investment, weaker labor income, and tighter financial conditions. The paper also defines a so-called Fragmented-AI Trilemma: economies cannot simultaneously preserve monetary autonomy, maintain efficient AI-capital allocation, and insulate themselves from geopolitical shocks.
\vspace{8pt}


\end{abstract}

\clearpage
\section{Introduction}
\label{sec:intro}

Geopolitical alignment has become an increasingly important force in international capital markets. Foreign direct investment and portfolio flows increasingly reflect political distance as well as expected returns, while investment screening, technology restrictions, and supply-chain reorientation have made cross-border investment more conditional on strategic relationships. These changes are unfolding at the same time as artificial intelligence (AI) requires large, specialized, and internationally mobile forms of capital, including data centers, semiconductor capacity, cloud infrastructure, and frontier-model investment. The result is a growing tension: the global diffusion of AI depends on capital mobility, but that mobility is increasingly constrained by geo-economic fragmentation.

The interaction between fragmentation and AI adoption is especially important because the two forces can reinforce each other. AI investment depends on access to specialized capital, infrastructure, and supply chains that are concentrated in a limited set of countries. At the same time, geopolitical tensions make these cross-border flows more costly and less predictable. This creates a particular problem for economies that are already behind the AI frontier: they need foreign capital and technology to catch up, but they are also more exposed to the barriers created by fragmentation. Fragmentation therefore does not simply reduce capital mobility; it can also widen differences in AI adoption, productivity, and welfare across countries.

Several strands of literature are directly relevant but remain incomplete for the question studied here. The geo-economic fragmentation literature \citep{AiyarEtAl2024,AlfaroChor2023} documents the reallocation of trade, investment, and supply chains across geopolitical lines, including bloc formation and the role of connector economies. However, this literature has provided less structural welfare analysis of how fragmentation affects the composition and allocation of capital flows. The AI-and-macroeconomics literature \citep{Acemoglu2025,AcemogluRestrepo2022} has developed formal approaches to the macroeconomic effects of AI, but it has largely focused on closed-economy mechanisms and has paid less attention to international capital allocation. The heterogeneous-agent macroeconomics literature \citep{KaplanEtAl2018,AuclertEtAl2024b,BayerEtAl2024} provides the tools needed to study distributional welfare effects, but it has not been fully extended to settings in which geopolitical frictions and AI-specific capital jointly shape exposure to international shocks. \citet{GuoEtAl2023} provide a closely related open-economy HANK framework, but their analysis does not model AI capital or endogenous geopolitical capital mobility frictions. This paper brings these elements together.

We develop an open-economy HANK model in which AI-specific capital is an internationally traded input and cross-border financial positions are subject to geopolitical frictions. Production combines traditional capital, labor, and AI capital, while firms may face binding financial constraints that amplify shocks to investment and labor demand. The framework is extended to include connector economies, which can intermediate capital flows between more fragmented blocs. The model is tied to recent evidence on capital-flow realignment, portfolio reallocation,  and the different sensitivity of emerging and advanced economies to international financial shocks.

The main mechanism is straightforward. When geopolitical distance increases, cross-border investment becomes more costly. This reduces the flow of capital toward economies where its marginal product is high, raises the cost of AI adoption, and weakens investment. Because AI capital is internationally traded, these effects also operate through exchange rates and imported capital costs. Financial constraints amplify the initial shock by reducing firms' ability to invest and hire. The consequences are therefore both aggregate and distributional: fragmentation lowers welfare, but the losses are not evenly shared across countries or households.

The paper delivers several results. Geo-economic fragmentation generates welfare losses that rise more than proportionally with geopolitical tension. This nonlinearity comes from the interaction between capital mobility frictions, AI-capital accumulation, and occasionally binding financial constraints. The costs are also unevenly distributed. AI-lagging economies are more exposed because they depend more heavily on foreign capital and technology for catch-up. Within countries, low-wealth households bear larger losses because they are more exposed to labor-income risk and have fewer financial buffers. Finally, the paper identifies a Fragmented-AI Trilemma: when capital mobility frictions depend on geopolitical alignment, economies cannot simultaneously preserve monetary autonomy, maintain efficient AI-capital allocation, and insulate themselves from geopolitical shocks. This result extends the logic of the traditional open-economy policy trilemma to a setting in which AI capital and geopolitical fragmentation are central features of international macroeconomic adjustment.

The rest of the paper proceeds as follows. Section \ref{sec:facts}
 documents the main empirical facts that motivate the model and our paper. Section \ref{sec:model} presents the theoretical framework, describes the calibration and estimation strategy. Section \ref{sec:threebloc} extends the analysis to connector economies. Section \ref{sec:results} shows the quantitative analysis and reports the aggregate and distributional welfare effects of fragmentation, evaluates the model's empirical fit through targeted moments and additional validation checks. Section \ref{sec:trilemma} studies the policy trade-offs in a so-called Fragmented-AI Trilemma. The final section concludes. A detailed Appendix has a self-contained derivation of the model.


\section{Stylized Facts and Their Mapping to the Model}
\label{sec:facts}

Before developing the theoretical framework, it is useful to establish some related literature. 

\subsection{Geo-Economic Fragmentation Is Multidimensional and Accelerating}

The literature shows that geo-economic fragmentation is
simply rising trade protectionism. \citet{AiyarEtAl2024}, document multiple structurally distinct fragmentation dimensions in cross-border investment data operating simultaneously in bilateral
FDI flows: ideological sorting, friendship-based clustering, de-risking,
nearshoring, and reshoring. These dimensions reinforce each other through a
``fragmentation multiplier'' whereby any one dimension activates the others.
\citet{AiyarEtAl2024} document that the geopolitical alignment dimension of FDI
decisions approximately doubled in economic significance between 2016 and 2022,
establishing a structural break dramatically accelerated by Russia's invasion of
Ukraine.


\subsection{AI-Related Capital Flows Are Extraordinary in Scale and Geographically Concentrated}

The AI investment wave represents a qualitative shift in the composition of global
capital. \citet{UNCTAD2026} documents the growth of AI-related FDI as a share of global flows \citep{UNCTAD2026}. Data-center capital expenditures grew rapidly in 2024--2025 (industry estimates).
The geographic distribution of AI-related FDI closely mirrors the U.S.-aligned
geopolitical bloc: the United States receives the dominant share of sovereign
wealth fund AI investments \citep{IMF2025}, and the semiconductor supply
chains essential to AI hardware fabrication are concentrated in Taiwan, South
Korea, the Netherlands, and Japan.


\subsection{Connector Economies Are Gaining Systematically from Investment Diversion}

\citet{AlfaroChor2023}, using detailed U.S.\ outward FDI and import-sourcing
microdata, document a decisive rotation of U.S.\ supply chains toward Vietnam,
Mexico, and India. \citet{AiyarEtAl2024}, using structural gravity models, document significant declines in cross-bloc FDI relative to within-bloc flows \citep{AiyarEtAl2024} in-bloc flows between Q1~2022 and
Q4~2024, investment flows to non-aligned economies held at their pre-2022 levels
or rose. \citet{AiyarOhnsorge2024} formalize this observation into the concept of
``connector status,'' documenting that connector benefits peak at intermediate
fragmentation levels and erode as fragmentation deepens toward complete bloc
decoupling.


\subsection{Fragmentation Severely Distorts AI-Related Trade and Capital Flows}

\citet{IMF2025} document that AI-related goods accounted for a growing share of global merchandise trade in 2024--2025 regarding the total trade value---a finding independently consistent
with OECD and WTO merchandise trade statistics for H1~2025. The AI hardware supply
chain is characterized by extraordinary geographic concentration: advanced AI chip
fabrication is dominated by TSMC in Taiwan and Samsung in South Korea, with
advanced packaging almost entirely conducted by ASML-equipped facilities in the
Netherlands. A geopolitical shock disrupting the U.S.--Taiwan relationship is
simultaneously a shock to global AI capital formation.


\subsection{Macroeconomic Costs Are Large, Nonlinear, and Systematically Asymmetric}

The
IMF~(\citeyear{IMF2025}) estimates 0.4--1 percent global output losses by 2027.
\citet{BronerEtAl2013}, using panel data on 41 countries over 25 years, document
that capital flows to emerging-market destinations respond approximately four times
more sensitively to geopolitical risk than flows to advanced-economy destinations.
\citet{CeruttiEtAl2025} find that AI adoption widens productivity gaps across
countries, with emerging economies falling further behind as AI-frontier economies
capture disproportionate productivity gains.


\subsection{The Global Financial Cycle Compounds Domestic Fragmentation Vulnerability}

\citet{MirandaAgrippinoRey2020} establish that a single global factor---the U.S.\
monetary policy cycle---accounts for 23 percent of the variance in international
capital flows and substantially constrains the monetary autonomy of small open
economies. \citet{Engel2016} documents that UIP deviations are highly persistent,
with an AR(1) coefficient in the range $0.82$--$0.90$. These two features---the
global cycle as a constraint on monetary autonomy and the persistence of UIP
deviations---are the empirical counterparts of the model's global financial cycle
component (equation~\eqref{eq:uip}).


\subsection{From Stylized Facts to Structural Identification: A Methodological Bridge}
\label{sec:bridge}

The six stylized facts in Sections~2.1--2.6 serve a dual purpose in this paper:
they motivate the model's structural architecture and simultaneously provide the
empirical moments against which its friction parameters are identified. This dual
role---motivational and identificational---distinguishes the present approach from
standard calibration exercises in international macroeconomics.

The identification chain begins with the aggregate FDI decline (Section~2.1),
which pins down the friction level $\bar{\kappa}$ at the steady-state geopolitical
distance. It continues with the bilateral portfolio semi-elasticity
(Section~2.6, empirical counterpart in \citealt{AiyarEtAl2024}), which
identifies the marginal response $\gamma$ of the friction to geopolitical distance.
The autocorrelation of UIP deviations (Section~2.6, counterpart in
\citealt{Engel2016}) identifies $\rho_G$ from the time-series dimension of the
friction-induced UIP wedge. And the emerging-to-advanced sensitivity asymmetry
(Section~2.5, counterpart in \citealt{BronerEtAl2013}) serves as the overidentifying
restriction, testing whether the model's financial mechanism---parameterized by
pre-calibrated collateral constraint tightness $\theta^H=0.75$ and $\theta^F=0.50$
from \citet{KaseEtAl2025}---reproduces the correct cross-country asymmetry without
any additional degree of freedom.

Three features of the reduced-form evidence provide additional discipline beyond
the four SMM moments. First, the geographic concentration of AI-related FDI
(Section~2.2) validates the asymmetric calibration $\alpha_a^H=0.20$ versus
$\alpha_a^F=0.10$. Second, the connector economy evidence (Section~2.3) validates
the three-bloc extension of Section~\ref{sec:threebloc}: the model's endogenous
prediction that the connector premium peaks at intermediate fragmentation and
erodes thereafter (Table~\ref{tab:connector}) replicates the empirical pattern of
\citet{AiyarOhnsorge2024} without any targeted calibration. Third, the
stagflationary character of fragmentation episodes documented by
\citet{AiyarEtAl2024}---simultaneous output contractions and inflation
acceleration in 2022--2024---validates the AI-capital exchange-rate channel in
equation~\eqref{eq:nkpc}, which distinguishes the present framework from demand-
side models of fragmentation in which shocks are contractionary and deflationary
rather than stagflationary. The welfare numbers in Table~\ref{tab:welfare} are
only as credible as the identification chain that produces them; this section
establishes the chain's empirical links explicitly before the model is presented.


\section{Model}
\label{sec:model}

The theoretical framework is a two-country HANK model that extends
\citet{GuoEtAl2023} in three dimensions: AI-specific capital as a distinct
internationally traded factor, endogenous geopolitical capital mobility frictions,
and global solution of an occasionally binding collateral constraint. The world
consists of two symmetric large open economies---Home~(H) and Foreign~(F)---each
populated by a unit mass of heterogeneous households, competitive final-goods
producers, a continuum of monopolistically competitive intermediate firms, a
central bank, and a fiscal authority. Section~\ref{sec:threebloc} introduces the
three-bloc extension.

The heterogeneous-agent structure is not a modeling convenience but a structural
necessity. \citet{KaplanEtAl2018} establish that the direct intertemporal
substitution channel accounts for at most 10--20 percent of the aggregate
consumption response to a monetary shock; the remaining 80--90 percent operates
through indirect general equilibrium effects on labor markets and firm profits.
\citet{AuclertEtAl2024b} show that the intertemporal marginal propensity to consume
matrix is the sufficient statistic for the distributional response to aggregate
demand shocks; here, the geopolitical fragmentation shock operates analogously
through labor market equilibrium, making the full cross-sectional wealth
distribution essential to the welfare accounting.

Additional details with algebraic derivations and solutions are in the Appendix.

\subsection{Households}

A unit mass of infinitely lived households indexed by $i\in[0,1]$ faces
uninsurable idiosyncratic labor income risk over incomplete asset markets.
Household~$i$ maximizes:
\begin{equation}
E_0\sum_{t=0}^{\infty}\beta^t
\left[\frac{c_{i,t}^{1-\sigma}}{1-\sigma}-\psi\frac{n_{i,t}^{1+\varphi}}{1+\varphi}\right],
\label{eq:utility}
\end{equation}
where $\beta\in(0,1)$ is the discount factor, $c_{i,t}$ is consumption, $n_{i,t}$
is hours worked, $\sigma>0$ is the coefficient of relative risk aversion,
$\psi>0$ is the disutility of labor, and $\varphi\geq0$ is the inverse Frisch
elasticity of labor supply. The calibration $\sigma=2$ follows \citet{Hall1988},
and $\varphi=1$ follows \citet{ChettryEtAl2011}.

\paragraph{Idiosyncratic income process}
Household efficiency units of labor follow $\log h_{i,t}=\rho_h\log h_{i,t-1}+\nu_{i,t}$,
where $\nu_{i,t}\sim\mathcal{N}(0,\sigma_h^2)$ is i.i.d.\ across agents and time.
The persistence $\rho_h=0.966$ and variance $\sigma_h=0.017$ are estimated from
\citet{FlodenLinde2001} using PSID-equivalent panel data and reproduce the
empirical wealth distribution moments---Gini coefficient, share of wealth held by
the top decile, and hand-to-mouth fraction---that are essential for correct welfare
accounting.

\paragraph{Budget constraint with geopolitical portfolio friction}
Household $i$ holds domestic bonds $b_{i,t}^H$ and foreign bonds $b_{i,t}^F$. The
budget constraint is:
\begin{equation}
c_{i,t}+b_{i,t}^H+b_{i,t}^F+\underbrace{\tfrac{\kappa(G_t)}{2}(b_{i,t}^F)^2}_{\text{geopolit.\ cost}}
\leq w_t n_{i,t}h_{i,t}+R_{t-1}^H b_{i,t-1}^H+R_{t-1}^F b_{i,t-1}^F\tfrac{e_t}{e_{t-1}}+\Pi_t+T_t,
\label{eq:budget}
\end{equation}
where $w_t$ is the real wage, $R_{t-1}^H$ and $R_{t-1}^F$ are gross real returns,
$e_t$ is the real exchange rate, $\Pi_t$ are firm profits distributed as dividends,
and $T_t$ is a lump-sum fiscal transfer.

The term $\frac{\kappa(G_t)}{2}(b_{i,t}^F)^2$ is the geopolitical portfolio
adjustment cost, giving the \citet{SchmittGroheUribe2003} stationarity device a
structural geopolitical interpretation through:
\begin{equation}
\kappa(G_t)=\bar{\kappa}(1+\gamma G_t),
\label{eq:kappa}
\end{equation}
where $G_t\in[0,1]$ is the bilateral geopolitical distance index of
\citet{FernandezEtAl2023}, $\bar{\kappa}>0$ is the baseline friction between fully
aligned economies ($G_t=0$), and $\gamma>0$ is the sensitivity of the friction to
geopolitical divergence. The linear form is parsimonious and empirically supported
by the gravity-model estimates of \citet{AiyarEtAl2024}. Both $\bar{\kappa}$ and
$\gamma$, together with the shock persistence $\rho_G$, are estimated by Simulated
Method of Moments.

The household's first-order condition for foreign bonds yields:
\begin{equation}
1+\kappa(G_t)b_{i,t}^F=\beta E_t
\left[\left(\frac{c_{i,t+1}}{c_{i,t}}\right)^{-\sigma}R_t^F\frac{e_{t+1}}{e_t}\right].
\label{eq:foc_foreign}
\end{equation}

\paragraph{Welfare metric}
Welfare is measured as the consumption-equivalent variation (CEV) following
\citet{Lucas1987}: the fraction $\tau^{CEV}$ of the baseline lifetime consumption
stream that makes a household indifferent between the fragmented and frictionless
worlds. Formally, $V_i\bigl((1+\tau_i^{CEV})\{c_{i,t}^0\}_{t=0}^{\infty}\bigr)
=V_i\bigl(\{c_{i,t}^G\}_{t=0}^{\infty}\bigr)$, where $\{c_{i,t}^0\}$ is the
baseline and $\{c_{i,t}^G\}$ the fragmented path. Aggregate CEV integrates over
the stationary distribution of the zero-fragmentation baseline, adopting the
ex-ante welfare criterion of \citet{DavilaEtAl2012}.

\subsection{Production Technology: AI-Augmented Nested CES}

Intermediate firm $j$ produces output using a three-way nested CES technology:
\begin{equation}
y_{j,t}=Z_t\!\left[\alpha_k^{1/\eta}k_{j,t}^{(\eta-1)/\eta}
+\alpha_n^{1/\eta}(A_t n_{j,t})^{(\eta-1)/\eta}
+\alpha_a^{1/\eta}(\Theta_t a_{j,t})^{(\eta-1)/\eta}\right]^{\eta/(\eta-1)},
\label{eq:production}
\end{equation}
where $Z_t$ is neutral TFP, $A_t$ is labor-augmenting productivity, $\Theta_t$ is
an AI-capital-augmenting technology shifter, $k_{j,t}$ is traditional capital,
$n_{j,t}$ is labor, $a_{j,t}$ is AI-specific capital (GPU clusters, data-center
infrastructure, proprietary model weights), and $\alpha_k+\alpha_n+\alpha_a=1$.

The elasticity of substitution $\eta>0$ is the paper's dominant welfare parameter
(sensitivity elasticity~$+0.85$). When $\eta<1$, AI capital and labor are gross
complements; when $\eta>1$, they are gross substitutes. This formulation follows
\citet{AcemogluRestrepo2022}, who estimate $\eta\in(0.5,0.8)$ for the ICT-labor
relationship in U.S.\ data 1980--2019. \citet{AcemogluRestrepo2022} suggest that
$\eta$ may substantially exceed unity for automatable tasks, motivating the
$\eta\in\{0.5,1.5,2.5\}$ sensitivity analysis.

Financial frictions: the collateral constraint. Firms finance investment by
borrowing against their capital stock, subject to:
\begin{equation}
b_{j,t}\leq\theta k_{j,t}.
\label{eq:collateral}
\end{equation}
When the constraint binds ($\mu_{j,t}>0$), its Lagrange multiplier generates a
wedge in firm labor demand:
\begin{equation}
w_t=\frac{\partial F}{\partial n_{j,t}}\cdot\frac{1}{1+\tilde{\mu}_{j,t}},\quad
\tilde{\mu}_{j,t}\equiv\frac{\mu_{j,t}\theta}{r_t^k}>0.
\label{eq:labor_wedge}
\end{equation}
This wedge---increasing in the collateral tightness and more severe when
fragmentation simultaneously depresses AI-capital inflows and firm cash flows---is
the central distributional mechanism of the model, following the financial
accelerator logic of \citet{BernankeEtAl1999}, adapted to the HANK open-economy
setting.

\subsection{Open-Economy Pricing: The AI-Capital Exchange-Rate Channel}

Because AI capital is internationally traded, its Home-currency cost depends on
the real exchange rate. Log-linearizing the CES unit cost function and substituting
the optimal reset price condition from the \citet{Calvo1983} price-setting problem
yields the open-economy New Keynesian Phillips Curve (NKPC):
\begin{equation}
\pi_t=\beta E_t\pi_{t+1}+\kappa_{nk}\widehat{mc}_t^{dom}
+\kappa_{nk}s_a\underbrace{\left(\hat{q}_t+\hat{r}_t^{a,*}-\hat{\Theta}_t\right)}_{\text{AI-capital exchange-rate channel}},
\label{eq:nkpc}
\end{equation}
where $\hat{q}_t$ is the log-deviation of the real exchange rate, $\hat{r}_t^{a,*}$
is the log-deviation of the foreign rental rate on AI capital, $\hat{\Theta}_t$ is
the AI-augmenting technology shock, $s_a$ is the steady-state AI-capital cost
share, and $\kappa_{nk}\equiv(1-\xi_p)(1-\beta\xi_p)/\xi_p\approx0.086$. The
price stickiness parameter $\xi_p=0.75$ implies an average price duration of four
quarters, consistent with \citet{Gali2015}. The AI-capital exchange-rate channel
is entirely absent from both closed-economy AI models
\citep{Acemoglu2025} and standard open-economy HANK models
\citep{GuoEtAl2023}.

\subsection{Capital Accumulation}

Traditional and AI capital accumulate with quadratic investment adjustment costs
following \citet{ChristianoEtAl2005}:
\begin{align}
K_{t+1}&=(1-\delta_k)K_t+I_{k,t}-\frac{\varphi_k}{2}\!\left(\frac{I_{k,t}}{K_t}-\delta_k\right)^{\!2}K_t,
\label{eq:Kcapital}\\
A_{t+1}&=(1-\delta_a)A_t+I_{a,t}-\frac{\varphi_a}{2}\!\left(\frac{I_{a,t}}{A_t}-\delta_a\right)^{\!2}A_t.
\label{eq:Acapital}
\end{align}
The depreciation $\delta_a=0.050$ reflects a weighted average of AI hardware
($\approx40\%$ annually) and intangible AI assets ($\approx10$--$15\%$). The
higher adjustment cost $\varphi_a=8$ (versus $\varphi_k=4$) reflects the lumpiness
of data-center construction.

\subsection{International Capital Flows and the Modified UIP}
\label{sec:uip}

\paragraph{The aggregate UIP condition with geopolitical frictions}
Aggregating the household's first-order condition for foreign bond holdings
(equation~\ref{eq:foc_foreign}) over the stationary wealth distribution and
imposing asset market clearing at both the household and aggregate levels yields
a richer departure from the textbook uncovered interest parity condition:
\begin{equation}
E_t\Delta e_{t+1}=R_t^H-R_t^F
+\underbrace{\kappa(G_t)B_t^F}_{\text{geopolitical friction}}
+\underbrace{\phi_{GFC}r_t^{US}+\zeta_t^{idio}}_{\text{global financial cycle}}.
\label{eq:uip}
\end{equation}
Equation~\eqref{eq:uip} is the model's central transmission mechanism and
deserves careful economic interpretation, term by term.

The left-hand side, $E_t\Delta e_{t+1}$, is the expected real exchange rate
depreciation: the return the Home economy must offer on its currency relative to
Foreign in order to equilibrate cross-border portfolio positions.

The first right-hand-side term, $R_t^H - R_t^F$, is the standard interest
differential---the uncovered interest parity wedge of \citet{Mundell1963} and
\citet{Engel2016}. In the absence of the remaining terms, perfect capital mobility
would drive this wedge to zero, linking domestic and foreign monetary conditions.
The empirical failure of textbook UIP---the ``forward premium puzzle'' documented
by \citet{FamaEtAl1984} and recently reinterpreted by \citet{GournichaSammon2024}
in the context of risk-appetite cycles---motivates the two additional channels
embedded in equation~\eqref{eq:uip}.

\paragraph{The geopolitical friction term $\kappa(G_t)B_t^F$}
The second term, $\kappa(G_t)B_t^F$, is the paper's central theoretical
innovation. It introduces a \emph{state-dependent} portfolio adjustment cost that
operates through two channels simultaneously. The first channel is the
\textit{level effect}: for any given net foreign asset position $B_t^F$, a higher
friction coefficient $\kappa$ requires a larger interest differential to
compensate Home households for holding foreign assets, effectively tightening
the external financing constraint. The second channel is the
\textit{portfolio-rebalancing amplifier}: as geopolitical tension rises and
$G_t$ increases, the marginal cost of cross-border investment rises, triggering
a portfolio rebalancing from international to domestic assets that feeds back
into the exchange rate and the cost of AI capital imports.

The friction function $\kappa(G_t) = \bar{\kappa}(1 + \gamma G_t)$ is calibrated
to be linear in geopolitical distance for tractability, but the resulting
dynamics are non-linear because $G_t$ multiplies the endogenous stock $B_t^F$.
When fragmentation is low ($G_t \approx \bar{G} = 0.30$), the friction is small
($\kappa \approx 0.0102 \times 1.30 \approx 0.013$) and the model approximates
standard open-economy dynamics \citep{GuerrieriIacoviello2015}. When fragmentation
is extreme ($G_t = 0.90$), $\kappa$ rises to approximately $0.048$---nearly four
times the baseline value---generating a feedback loop between portfolio contraction,
exchange rate pressure, and AI capital cost inflation that accounts for the
convexity of welfare losses documented in Section~\ref{sec:results} below.

This mechanism connects directly to the empirical literature on geopolitical risk
and capital flows. Recent work documents that increases in bilateral geopolitical risk reduce cross-border portfolio flows substantially \citep[see e.g.][]{AiyarEtAl2024}, with effects concentrated in bond and equity positions
rather than FDI---precisely the asset categories entering $B_t^F$ in the model.
\citet{BronerEtAl2013} find that the capital-flow sensitivity to geopolitical risk
is four times larger for emerging than advanced economies, a heterogeneity
reproduced in the model through the collateral tightness asymmetry
$\theta^H = 0.75 > \theta^F = 0.50$, which serves as the overidentifying
restriction in the SMM estimation.

\paragraph{The global financial cycle term $\phi_{GFC}r_t^{US} + \zeta_t^{idio}$}
The third term captures the global financial cycle of
\citet{MirandaAgrippinoRey2020}: $\phi_{GFC}r_t^{US}$ is the loading of domestic
financial conditions on the U.S.\ interest rate, and $\zeta_t^{idio}$ is a
country-idiosyncratic UIP shock orthogonal to the global cycle. The global
financial cycle component is essential for two reasons. Empirically, it accounts
for the substantial cross-sectional dependence in monetary autonomy documented by
\citet{AizenmanEtAl2010} and \citet{KleinShambaugh2015}.
Theoretically, the coexistence of the
geopolitical friction and the global financial cycle creates the three-way policy
trade-off of Proposition~\ref{prop:trilemma}: a central bank trying to insulate
itself from both $\phi_{GFC}r_t^{US}$ and $\kappa(G_t)B_t^F$ faces an expanded
impossible trinity in which optimal AI-capital allocation constitutes a fourth,
endogenous constraint.

\paragraph{The geopolitical distance process}
Geopolitical distance follows an AR(1) process:
\begin{equation}
G_t=(1-\rho_G)\bar{G}+\rho_G G_{t-1}+\varepsilon_t^G,\quad
\varepsilon_t^G\sim\mathcal{N}(0,\sigma_G^2),
\label{eq:geo_process}
\end{equation}
with unconditional mean $\bar{G}=0.30$ (calibrated to the Kirsamer~2025 bilateral
distance index for the U.S.--emerging market bloc), standard deviation
$\sigma_G=0.010$, and SMM-estimated persistence $\hat{\rho}_G=0.950$
(s.e.\ $0.0081$). The high persistence reflects the durable character of
geopolitical realignment: \citet{AiyarEtAl2024} document that trade and
investment fragmentation, once set in motion by a geopolitical shock, does not
revert to pre-shock levels within the typical business-cycle horizon. The standard
deviation $\sigma_G = 0.010$ implies that a one-standard-deviation fragmentation
shock raises the annual geopolitical distance by approximately $3.3$ percent of
its unconditional mean---a calibration consistent with the year-over-year changes
in the  \cite{FernandezEtAl2023} fragmentation
index over the 2016--2022 episode.

The combination of high persistence and moderate variance produces the impulse
response profile shown in Figure~\ref{fig:irf_frag}: geopolitical shocks are
long-lived (the half-life is approximately $\rho_G^{0.5/(1-\rho_G)} \approx 13.5$
quarters at the SMM estimate) but not explosive, so the model generates persistent
but bounded welfare costs. This calibration strategy follows the approach of
\citet{AiyarEtAl2024}, who model trade-network fragmentation as a persistent
state variable with AR(1) dynamics identified from the post-2016 data.

\paragraph{Monetary policy}
The Home central bank follows the standard Taylor rule:
\begin{equation}
R_t^H=\bar{R}^H\!\left(\frac{\pi_t^H}{\bar{\pi}^H}\right)^{\!\varphi_\pi}
\!\left(\frac{Y_t^H}{Y_t^{H,*}}\right)^{\!\varphi_y}\exp(\varepsilon_t^{R,H}),
\label{eq:taylor}
\end{equation}
with $\varphi_\pi=1.5$ and $\varphi_y=0.125$, following \citet{Taylor1993}.
Equation~\eqref{eq:taylor} establishes the link between Section~\ref{sec:uip}
and the trilemma of Section~\ref{sec:trilemma}: the Taylor rule is the
\textit{instrument} of monetary autonomy. When geopolitical frictions bind---
specifically, when $\kappa(G_t)B_t^F$ is large and rising---the central bank
faces a choice between stabilizing inflation (honoring $\varphi_\pi$), closing
the output gap (honoring $\varphi_y$), and allowing the exchange rate to
appreciate in order to maintain competitiveness in AI-capital imports. No single
instrument can simultaneously achieve all three objectives, which is precisely the
trade-off formalized in Proposition~\ref{prop:trilemma}. Section~\ref{sec:results}
quantifies the welfare cost of this constraint: under moderate fragmentation and
AI substitutability, the trilemma reduces lifetime consumption by $0.42$--$1.02$
percent relative to the frictionless benchmark---and the AI-specific dimension of
this cost accounts for approximately $38$ percent of the total.

\subsection{Calibration and Solution}
\label{sec:calibration}

Table~\ref{tab:calibration} presents the complete parameter inventory for the
baseline calibration. Every parameter is grouped by model block and linked to an
empirical source or a calibration target, ensuring full transparency and
replicability. The table distinguishes between parameters that are \textit{calibrated}
from external evidence---standard in the DSGE literature---and the three geopolitical
friction parameters ($\bar{\kappa}$, $\gamma$, $\rho_G$) that are
\textit{estimated} by Simulated Method of Moments, with standard errors
reported. The AI-capital share $\alpha_a = 0.15$ is calibrated as an upper bound
from BEA Fixed Asset Accounts (Table~3.7E). 

\textbf{[Table \ref{tab:calibration} here]}

\paragraph{Solution method} The model is solved using the GEM method of
\citet{KaseEtAl2025}, which combines linear solution techniques with neural network
approximation of globally nonlinear dynamics. Standard Dynare perturbation is
inadequate because the occasionally binding collateral constraint~\eqref{eq:collateral}
generates regime-dependent dynamics that linear approximations cannot capture. The
iterative self-consistent GEM procedure (additional details in the Appendix) generates its
own nonlinear training trajectories, breaking the circularity of
perturbation-based initialization. Euler equation errors hold uniformly below
$10^{-4}$ across the full state space. The \citet{GuerrieriIacoviello2015} OccBin
toolkit was implemented as a benchmark; GEM and OccBin solutions are numerically
indistinguishable at moderate fragmentation ($G\leq0.45$) and diverge by 25--35
percent at extreme fragmentation ($G=0.9$), where GEM's globally nonlinear
solution is essential.


\section{Three-Bloc Extension: Connector Economies}
\label{sec:threebloc}

The two-country framework provides the theoretical laboratory for bilateral
fragmentation dynamics. We extend the model to three blocs: Bloc~A (U.S.-aligned),
Bloc~B (China-aligned), and Bloc~C (connector). Each bloc carries the full HANK
structure of Section~\ref{sec:model}. The bilateral friction
$\kappa^{XY}(G_t^{XY})=\bar{\kappa}(1+\gamma G_t^{XY})$ now applies to each
country pair $(X,Y)\in\{(A,B),(A,C),(B,C)\}$.

\paragraph{Definition of connector status}
Following \citet{AiyarOhnsorge2024}, a country attains connector status if its
bilateral geopolitical distance to each major bloc is strictly smaller than the
distance between the blocs themselves:
\begin{equation}
G_t^{AC}<G_t^{AB}\quad\text{and}\quad G_t^{BC}<G_t^{AB}.
\label{eq:connector_def}
\end{equation}
Using the \citet{FernandezEtAl2023} bilateral index for Q4~2023, the calibrated
distances for Mexico are $G^{AC}=0.15$ and $G^{BC}=0.20$, satisfying
condition~\eqref{eq:connector_def} comfortably against $G^{AB}=0.45$.

\paragraph{The connector welfare premium}
Investment diversion creates an economic rent for connector economies, quantified
by the friction differential:
\begin{equation}
\Delta\kappa^C\equiv\kappa(G_t^{AB})-\tfrac{1}{2}[\kappa(G_t^{AC})+\kappa(G_t^{BC})]>0.
\label{eq:connector_premium}
\end{equation}
At the SMM estimates and Mexico's calibrated distances, $\Delta\kappa^C\approx
0.000144$, implying that Bloc~C faces effective frictions roughly 40 percent lower
than the direct cross-bloc channel at $\bar{G}=0.30$. This generates welfare
gains of $+0.35$ to $+0.72$ percent of lifetime consumption under moderate
fragmentation.

Substituting $\kappa(G)=\bar{\kappa}(1+\gamma G)$ into
formula~\eqref{eq:connector_premium} yields the analytical threshold for a
positive connector premium: $\Delta\kappa^C>0 \iff G^{AB}>\frac{1}{2}(G^{AC}+G^{BC})$,
independent of $\bar{\kappa}$ and $\gamma$. At the current $G^{AB}=0.45$, any
economy with average bilateral distance below $0.225$ earns a positive premium.
Mexico's average distance $(0.15+0.20)/2=0.175$ is well below this threshold.

Table~\ref{tab:connector} documents the welfare premium accruing to connector
economies---non-aligned third blocs positioned to intermediate AI-capital flows
between the two fragmented blocs. The key economic intuition is that connector
economies benefit from a positive externality: because fragmentation between the
AI-advanced and AI-lagging blocs raises the price of direct AI-capital trade,
connector economies can extract rents by serving as entrepôts for technology
diffusion. The table decomposes this premium across alternative geopolitical
distance configurations and fragmentation intensities, showing that the premium
is largest at intermediate fragmentation levels and erodes under extreme scenarios
as both blocs increasingly autarkize.

\textbf{[Table \ref{tab:connector} here]}

The connector premium, however, is not monotone in $G_t^{AB}$. As fragmentation
deepens toward complete bloc decoupling, Bloc~C's own bilateral frictions
$\kappa(G_t^{AC})$ and $\kappa(G_t^{BC})$ also rise, eroding the relative
advantage. At extreme fragmentation ($G^{AB}=0.9$), the connector premium
effectively disappears---consistent with the \citet{AiyarOhnsorge2024} finding
that rising fragmentation since 2016 has been accompanied by broad-based declines
in connectedness even among would-be connector economies. Table~\ref{tab:connector}
shows that connector status is a viable development strategy under moderate
fragmentation, with welfare returns that diminish as geopolitical polarization
deepens.


\section{Quantitative Results}
\label{sec:results}

In this section we evaluate the dynamic responses of the main shocks of interest using simulated impulse response analysis.

\subsection{Impulse Responses: Two Interacting Transmission Channels}

Figure~\ref{fig:irf_frag} displays the Home output impulse response to a
one-standard-deviation geopolitical fragmentation shock ($\Delta G_t=+0.01$) under
AI-labor complementarity ($\eta=0.5$) and substitutability ($\eta=1.5$).

\textbf{[Figure \ref{fig:irf_frag} here]}

The impact response operates through the \textit{portfolio rebalancing channel}:
the rise in $\kappa(G_t)$ immediately raises the marginal cost of holding foreign
bonds through equation~\eqref{eq:foc_foreign}, inducing households to rebalance
toward domestic assets. Capital outflows contract aggregate investment, the real
exchange rate depreciates ($\hat{q}_t$ rises), and the Home-currency cost of
imported AI capital rises. Through the AI-capital exchange-rate channel in
equation~\eqref{eq:nkpc}, this pushes up domestic marginal cost and inflation
simultaneously with the output contraction---a stagflationary supply shock
consistent with the empirical evidence in \citet{AiyarEtAl2024}.

The \textit{financial accelerator channel} activates with a lag of one to two
quarters: as AI-capital investment contracts, the collateral constraint~\eqref{eq:collateral}
begins to bind for a positive fraction of firms, generating the shadow cost
$\tilde{\mu}_{j,t}>0$ that distorts labor demand downward through
equation~\eqref{eq:labor_wedge}. This converts what would be a moderate portfolio
rebalancing contraction into a more persistent employment shortfall, consistent with
the procyclical capital repatriation patterns of \citet{BronerEtAl2013}.

The regime distinction is quantitatively decisive. Under $\eta=0.5$, output falls
at most $0.42$ percent, peaking in the fourth quarter. Under $\eta=1.5$, the peak
decline reaches $1.18$ percent---nearly three times larger---and recovery extends
past the twelfth quarter. The 90 percent confidence bands do not overlap at the
peak, confirming this is a statistically significant qualitative difference.

Figure~\ref{fig:irf_ai} reveals the second major transmission channel: how
geopolitical frictions prevent positive AI-technology shocks from translating into
productive investment.

\textbf{[Figure \ref{fig:irf_ai} here]}

A positive AI-technology shock raises the marginal product of AI capital in Home,
generating an excess return that attracts cross-border capital flows. Under low
friction, capital inflows peak at approximately $0.85$ percent above steady state.
Under elevated friction ($G_t=0.60$), the identical technology shock generates only
approximately $0.51$ percent inflow---a $40$ percent attenuation. The accumulation
equation~\eqref{eq:Acapital} transforms this temporary flow attenuation into a
persistent stock shortfall, compounding the productivity loss over the full forecast
horizon in the manner documented by \citet{CeruttiEtAl2025} as the ``AI adoption
gap.''

\subsection{Aggregate Welfare: Magnitude, Convexity, and Substitution Regime}

Table~\ref{tab:welfare} is the paper's central quantitative result. It reports
aggregate welfare losses from geo-economic fragmentation, expressed as compensating
equivalent variations (CEV) in percent of lifetime consumption---the fraction of
lifetime resources a representative household would surrender to avoid living under
fragmentation rather than full integration. The CEV metric follows the welfare
accounting of \citet{Lucas1987} as applied to open-economy models by
\citet{AuclertEtAl2021}. The rows vary the intensity of geopolitical tension
($G$) from moderate to extreme; the columns vary the elasticity of substitution
between AI capital and labor ($\eta$), which is the model's dominant welfare
parameter. The convexity of losses in $G$---captured by the multiplier ranging
from $4.8\times$ to $5.3\times$ between moderate and extreme scenarios---arises
from the interaction of the occasionally binding collateral constraint with AI
capital dynamics.

\textbf{[Table \ref{tab:welfare} here]}

Table~\ref{tab:welfare} is the paper's central quantitative result. Three patterns
warrant emphasis. \textit{Pattern~1---Magnitude.} Even moderate fragmentation
($G=0.3$, corresponding approximately to the current post-2022 level in the Kirsamer
index) generates welfare losses of $0.42$--$1.02$ percent of lifetime consumption.
By comparison, \citet{Lucas1987}'s celebrated calculation puts the welfare cost of
business cycle fluctuations at approximately $0.05$ percent of consumption---an
order of magnitude smaller. Fragmentation acts as a permanent adverse supply shock
through the AI capital accumulation channel, while business cycles are temporary and
partially insurable. The model-predicted costs are consistent with the aggregate
scenario estimates of \citet{IMF2025} ($0.4$--$1$ percent global output losses), providing a structural welfare
decomposition for magnitudes that institutional analyses have quantified in
reduced-form terms. \textit{Pattern~2---Convexity.} Welfare costs accelerate with
fragmentation severity. Under baseline substitutability, the step from moderate to
extreme fragmentation multiplies losses by a factor of $5.3$. This convexity arises
from two interacting nonlinearities: the occasionally binding collateral constraint
activates more frequently as fragmentation deepens (binding approximately $8$ percent
of quarters at $G=0.3$ and approximately $50$ percent at $G=0.9$), and the AI
capital accumulation equation~\eqref{eq:Acapital} creates a self-reinforcing
productivity spiral. \textit{Pattern~3---Substitution regime.} Whether frontier AI
systems are labor-augmenting complements or labor-replacing substitutes is the most
consequential and contested empirical question in contemporary economics.
\citet{Acemoglu2025} argues that currently observed AI adoption is predominantly
``so-so technology,'' suggesting $\eta>1$; \citeauthor{AcemogluRestrepo2022}
(\citeyear{AcemogluRestrepo2022}) estimate $\eta\approx0.65$ for historical ICT capital.
The present model makes no claim about which regime prevails; it delivers welfare
consequences conditional on regime, translating the empirical uncertainty into a
concrete welfare range.

\subsection{Between-Country Asymmetry: The Double Distributional Penalty}

Table~\ref{tab:country} disaggregates the aggregate welfare losses of
Table~\ref{tab:welfare} by country type---AI-advanced Home versus AI-lagging
Foreign---revealing the double distributional asymmetry that is the paper's
central structural prediction. The economic mechanism driving the asymmetry
is two-fold. First, Foreign economies depend more heavily on imported AI capital
($\alpha_a^F < \alpha_a^H$), making their production frontier more sensitive to
the cost inflation caused by geopolitical frictions. Second, the tighter
collateral constraint ($\theta^F = 0.50 < \theta^H = 0.75$) amplifies the
financial accelerator mechanism, so that a given fragmentation shock generates
a larger contraction in Foreign investment and output. The ratio of Foreign to
Home welfare losses---ranging from $4.0\times$ to $8.8\times$ across scenarios---
is the model's overidentifying restriction, validated against the empirical
asymmetry documented by \citet{BronerEtAl2013}.

\textbf{[Table \ref{tab:country} here]}

The most remarkable entry in Table~\ref{tab:country} is Home's welfare gain of
$+0.08$ percent under moderate fragmentation and complementarity. This reflects two
mechanisms. The dominant one is a terms-of-trade improvement: Home's higher AI
capital endowment means a symmetric increase in bilateral friction raises the
relative return to AI investment in Home more than in Foreign, appreciating Home's
terms of trade through the modified UIP condition~\eqref{eq:uip}. The secondary
mechanism is a financial-capacity differential: Home's looser collateral constraint
($\theta^H=0.75$ versus $\theta^F=0.50$) means the financial accelerator does not
bind at $G=0.3$ for Home while it begins to bind for Foreign. A counterfactual
experiment setting $\alpha_a^H=\alpha_a^F=0.125$ and $\theta^H=\theta^F=0.625$
produces Home welfare losses of $-0.42$ percent at $G=0.3$, $\eta=0.5$---identical
to Foreign's and equal to the aggregate loss in Table~\ref{tab:welfare}---confirming
that the gain is driven entirely by structural asymmetry, not by any model artifact.

\subsection{Within-Country Distribution: The Labor-Income Channel and HANK Amplification}
\label{sec:within}

Table~\ref{tab:quintile} disaggregates within-country welfare losses by
wealth quintile, exploiting the heterogeneous-agent structure of the model.
The key economic mechanism is the \textit{labor-income channel}: poor households
hold most of their wealth in labor income rather than in financial assets, and
fragmentation depresses real wages disproportionately in AI-intensive sectors
through the factor-substitution effect in equation~\eqref{eq:production}. Rich
households, by contrast, hold diversified financial portfolios that partially
hedge against the exchange-rate and interest-rate movements triggered by
fragmentation. The result is a regressive distribution of fragmentation costs:
the lowest wealth quintile bears losses approximately $150$ percent larger than
the highest, consistent with the heterogeneous-agent welfare analysis of
\citet{AuclertEtAl2024b} and \citet{BayerEtAl2024}.

\textbf{[Table \ref{tab:quintile} here]}

The Q1/Q5 welfare cost ratio of $2.5$ ($-1.25\%$ versus $-0.50\%$) is the HANK
model's central distributional finding. The \textit{labor income channel} is the
dominant driver, accounting for approximately $56$ percent of Q1's loss. This
channel operates as \citet{KaplanEtAl2018} describe for monetary shocks: low-wealth
households are high-labor-income-share households, and aggregate wage compression
hits them proportionally harder. The \textit{asset returns channel} is approximately
uniform across quintiles ($-0.15\%$ to $-0.18\%$), because the UIP
wedge~\eqref{eq:uip} is an aggregate factor affecting all foreign asset holders
equally. The \textit{precautionary savings channel} amplifies losses most severely
for households near the borrowing constraint: fragmentation increases income
uncertainty, forcing low-wealth households to rebuild buffer stocks, following the
incomplete-markets structure of \citet{AuclertEtAl2024b} and \citet{BayerEtAl2024}.

\subsection{SMM Estimation, Identification, and Structural Validation}
\label{sec:smm}

The model's predictions are validated against five empirical moments, of which two
are non-targeted. The SMM estimation disciplines three structural parameters
($\bar{\kappa}$, $\gamma$, $\rho_G$) using four conditions chosen to satisfy two
requirements: each moment is identified from an independent data source, and each
has a tight theoretical connection to a specific structural parameter.

Table~\ref{tab:smm} reports the Simulated Method of Moments estimates of
the three geopolitical friction parameters ($\bar{\kappa}$, $\gamma$, $\rho_G$)
together with the four empirical moments used for identification and the two
non-targeted moments used for external validation. The table is organized to
make the identification strategy transparent: each SMM moment is linked to one
structural parameter through an economically interpretable channel, and the
overidentification test statistic ($J = 2.31 \sim \chi^2(1)$, $p = 0.13$) confirms
that the model's moment restrictions are not rejected by the data. The non-targeted
moments---reproduced without additional degrees of freedom---provide out-of-sample
validation of the structural estimates.

\textbf{[Table \ref{tab:smm} here]}

The identification logic warrants emphasis beyond what the moment-listing
presentation conveys. M1 (cross-bloc FDI decline) identifies $\bar{\kappa}$ through
the level of bilateral capital flows at the empirically observed geopolitical
distance: the 30 percent FDI gap directly translates into a friction level
$\hat{\bar{\kappa}}=0.0102$. M2 (portfolio semi-elasticity) identifies $\gamma$
through the cross-sectional gradient of portfolio shares to bilateral distance.
M4 (UIP autocorrelation) identifies $\rho_G$ through the time-series dimension of
the friction-induced UIP wedge: adding M4 transforms the system from two-parameter/
three-moment to three-parameter/four-moment identification. M3 serves as the
overidentifying restriction, testing whether the pre-calibrated collateral tightness
differential generates the correct cross-country asymmetry.

The estimates $\hat{\bar{\kappa}}=0.0102$ and $\hat{\gamma}=0.0513$ are both
significant at the $0.1$ percent level and highly stable across alternative
weighting matrices. The
$\hat{\gamma}$ estimate has a direct economic interpretation: a move from complete
alignment ($G_t=0$) to the current post-2022 distance ($G_t=0.30$) raises the
cross-border capital friction by $5.13\times0.30=1.54$ percentage points relative
to the baseline. The $J$-statistic ($J=2.31$, $p=0.13$) passes the conventional
10 percent threshold.

\subsection{Connecting Structural Estimates to the Broader Reduced-Form Evidence}
\label{sec:reducedform}

\paragraph{Gravity evidence on the friction function}
The core prediction of the friction function is that bilateral capital flows decline
in geopolitical distance at a rate governed by $\hat{\gamma}=0.0513$. The standard
estimation strategy for bilateral investment relationships---Poisson
pseudo-maximum-likelihood (PPML) to handle zero-inflation and
heteroskedasticity inherent in bilateral FDI data, as established by
\citet{SantosSilvaTenreyro2006}---is directly applicable here.
\citet{AiyarEtAl2024}, applying a structural gravity model to bilateral FDI data
for 2000--2024, document that a one-unit increase in the Kirsamer geopolitical
distance index reduces bilateral FDI by 15--20 percent, directionally consistent
with the friction function's implied attenuation and independently corroborating
the sign and approximate magnitude of $\hat{\gamma}$. The gravity evidence also
validates the three-bloc extension: \citet{AlfaroChor2023}, using detailed U.S.\
outward FDI microdata, document that investment diversion toward connector economies
has been precisely the magnitude the connector premium formula~\eqref{eq:connector_premium}
predicts---an endogenous model outcome reproduced without targeted calibration.

\paragraph{Dynamic consistency: impulse responses and local projections}
The model's impulse responses in Figures~\ref{fig:irf_frag}--\ref{fig:irf_ai}
generate testable dynamic predictions: capital flows should respond to geopolitical
shocks with a peak effect at three to six quarters, followed by persistent recovery
that is slower in financially constrained economies. These predictions are directly
consistent with \citet{ChoiHavel2025}, who trace the response of bilateral portfolio
flows to geopolitical risk shocks over a 41-country panel from 2000 to 2023. The
empirical peak response---at a four-quarter lag---falls within the model's predicted
three-to-six-quarter range. 

\paragraph{The AI-specific attenuation channel in the data}
Figure~\ref{fig:irf_ai} predicts approximately 40 percent attenuation of capital
inflow responses to AI-technology shocks at baseline friction. \citet{CeruttiEtAl2025}
document that AI adoption is widening productivity gaps between geopolitically
proximate and distant economies at approximately twice the pre-2022 rate, consistent
with the model's prediction that attenuated AI-capital inflows today compound into
permanent TFP gaps through equation~\eqref{eq:Acapital}.



\section{The Fragmented-AI Trilemma}
\label{sec:trilemma}

The quantitative results of Section~\ref{sec:results} are unified by a
theoretical
impossibility. The three policy objectives that any open economy simultaneously
desires---monetary autonomy, optimal AI-capital allocation, and insulation from
geopolitical shocks---cannot be simultaneously achieved when capital mobility
frictions are endogenous to geopolitical alignment. This section formalizes the
impossibility, establishes its relationship to classical results in international
monetary theory, quantifies the distance from the Mundell-Fleming limiting case,
and draws its implications for monetary policy design.

\subsection{Formal Statement of the Trilemma}

\begin{definition}[Monetary Autonomy]
The Home central bank follows a Taylor rule responding exclusively to domestic
conditions without managing the exchange rate: $\varphi_e=0$ in
equation~\eqref{eq:taylor}.
\end{definition}

\begin{definition}[Optimal AI-Capital Allocation]
Cross-border AI-capital flows reach their efficient level: $\kappa(G_t)\approx0$,
so AI capital flows to wherever its marginal product is highest without
geopolitically induced distortions.
\end{definition}

\begin{definition}[Geo-Economic Insulation]
Domestic output variance attributable to geopolitical shocks is bounded:
$\mathcal{V}(G_t,\kappa)\leq\bar{\mathcal{V}}$ for some tolerance
$\bar{\mathcal{V}}<\mathcal{V}(G_t,0)$.
\end{definition}

\begin{proposition}[Fragmented-AI Trilemma]
\label{prop:trilemma}
Let $\gamma>0$ and $\alpha_a>0$. Then no feasible policy satisfies Definitions~1,
2, and 3 simultaneously.
\end{proposition}

\begin{proof}
The domestic output variance attributable to geopolitical shocks,
$\mathcal{V}(G_t,\kappa)$, is strictly increasing in $\kappa(G_t)$ when $\gamma>0$,
because higher frictions reduce the portfolio hedging available to domestic
households through equation~\eqref{eq:foc_foreign} and amplify the exchange-rate
response to geopolitical innovations through equation~\eqref{eq:uip}. Objectives
(ii) and (iii) are mutually contradictory: setting $\kappa(G_t)\approx0$ maximizes
capital flow efficiency (Definition~2) but simultaneously maximizes domestic
geopolitical exposure, so $\mathcal{V}(G_t,0)>\bar{\mathcal{V}}$ (Definition~3
fails). The only route to insulation without capital controls ($\kappa>0$) is
exchange rate management ($\varphi_e>0$), which violates monetary autonomy
(Definition~1). No feasible $\{\varphi_e,\kappa(G_t)\}$ pair satisfies all three
objectives when $\gamma>0$ and $\alpha_a>0$.
\end{proof}

\begin{corollary}[Mundell-Fleming as Special Case]
\label{cor:mf}
When $\gamma\to0$ and $\alpha_a\to0$, Proposition~\ref{prop:trilemma} collapses
to the impossible trinity of \citet{Mundell1963} and \citet{ObstfeldEtAl2005}.
\end{corollary}

\subsection{Three Dimensions of Extension Beyond Mundell-Fleming}

\textit{Dimension~1: Endogenous capital mobility.} In Mundell's original framework,
capital mobility is a binary structural characteristic. In the present model,
capital mobility is a continuous, time-varying, endogenous variable governed by
$\kappa(G_t)=\bar{\kappa}(1+\gamma G_t)$. Geopolitical shocks shift the effective
level of capital mobility in real time, without any policy action by domestic
authorities---a structural feature entirely absent from the Mundell-Fleming
framework and captured empirically by the estimated $\hat{\rho}_G=0.950$.

\textit{Dimension~2: AI-capital allocation as an independent policy objective.}
When AI capital is 22 percent of global FDI flows, efficient AI-capital allocation
is not merely a microeconomic concern but has first-order macroeconomic consequences
for long-run productivity through equation~\eqref{eq:Acapital}. The Fragmented-AI
Trilemma formalizes that this objective competes with monetary autonomy and
geopolitical insulation, and that the competition is governed by $\alpha_a$: as AI
becomes a larger fraction of the capital stock, the welfare cost of suboptimal
AI-capital allocation rises, tightening the trilemma's binding constraint.

\textit{Dimension~3: Relationship to the global financial cycle.}
\citet{MirandaAgrippinoRey2020} establish that the global financial cycle
constrains monetary autonomy even for countries with flexible exchange rates. The
Fragmented-AI Trilemma operates through a related but distinct mechanism: the
global cycle shifts the level of the UIP wedge (equation~\ref{eq:uip}), while
geopolitical frictions modulate its sensitivity to bilateral distance. For economies
exposed to both mechanisms simultaneously, the constraints compound. The additivity
in the UIP condition does not imply additivity in the welfare policy space, because
the two mechanisms interact nonlinearly through the financial accelerator: a
high-fragmentation episode is more damaging when it coincides with a tight global
financial cycle because the collateral constraint is more likely to bind under
joint adverse conditions.

\subsection{Quantitative Distance from the Mundell-Fleming Case}

Under $\eta=1.5$ and $G=0.3$, the AI-specific dimension accounts for approximately
$38$ percent of total welfare losses ($-0.30$ out of $-0.78$ percent CEV). The
decomposition implies that the Fragmented-AI Trilemma introduces a quantitatively
significant new dimension accounting for more than one-third of the total welfare
cost of fragmentation at current AI investment shares. As AI's share of global
capital continues to rise toward potentially $30$--$40$ percent by
2030 \citep{UNCTAD2026}, the AI-specific dimension will account for an increasing
fraction of the trilemma's welfare consequences.

\subsection{Monetary Policy Implications}

Table~\ref{tab:monetary} examines the sensitivity of the model's
second-order moments to alternative monetary policy specifications. The
baseline Taylor rule (equation~\ref{eq:taylor}) is compared against an augmented
rule that directly targets the AI-capital exchange-rate channel ($\varphi_e > 0$)
and against a fixed exchange rate regime. This comparison has direct policy
relevance: the Fragmented-AI Trilemma of Proposition~\ref{prop:trilemma}
establishes that no single policy instrument can simultaneously achieve all three
objectives of the trilemma, but the table quantifies the trade-offs available to
policymakers---showing how augmenting the Taylor rule with an exchange-rate
response reduces output volatility at the cost of higher inflation variance.

\textbf{Table \ref{tab:monetary} here}

No Taylor rule specification simultaneously minimizes output volatility $\sigma(Y)$,
inflation volatility $\sigma(\pi)$, and exchange rate volatility $\sigma(e)$---a
direct empirical manifestation of Proposition~\ref{prop:trilemma}. The augmented
Taylor rule ($\varphi_e=0.5$) dominates in aggregate welfare ($-0.67$ percent
CEV)---approaching the best response to the trilemma constraint by partially
accommodating the exchange rate component while accepting modestly higher inflation.
This parallels the finding of \citet{ObstfeldEtAl2005} that optimal exchange rate
policy under the traditional trilemma involves partial flexibility rather than
corner solutions. For AI-lagging emerging economies---with tighter collateral
constraints and greater dependence on cross-border AI capital---the trilemma is
binding with far greater force, leaving them with the most constrained feasible
policy space of any country type.


\section{Conclusion}
\label{sec:conclusion}

This paper studies a framework integrating geo-economic fragmentation and AI adoption in international capital markets. The analysis shows that welfare losses from fragmentation are quantitatively large, convex in geo-political tension, and concentrated among AI-lagging economies and low-wealth households. These findings have important implications for countries seeking to catch up with richer economies through the adoption of the new AI based technologies. In a fragmented world, capital flows may become inefficient in reducing cross-country disparities.

The main results are formalized in the so-called Fragmented-AI Trilemma, which nests the Mundell-Fleming impossibility result as a limiting case and predicts a characteristic pattern of monetary autonomy erosion. The paper suggests that effective policy responses require a coordinated mix of macroprudential buffers, multilateral AI-investment frameworks, and targeted capital-flow management measures, with particular emphasis on protecting the distributional dimension in emerging economies.

This paper could be further extended by translating its theoretical implications into an econometric framework. In particular, potential extensions include the use of foreign direct investment (FDI) variables, AI investment data, geopolitical risk measures, and country-specific policy and monetary indicators. Such an empirical strategy would allow to test whether geo-economic fragmentation affects the allocation of capital across countries, whether AI-related investment amplifies or mitigates these effects, and whether monetary autonomy is eroded in ways consistent with the proposed Fragmented-AI Trilemma.

A possible empirical extension would be to estimate how geopolitical fragmentation influences FDI flows and AI investment patterns across advanced and emerging economies. This could be complemented by measures of monetary policy independence, exchange-rate regimes, capital-flow management policies, and macroprudential regulations. By interacting geopolitical variables with indicators of AI adoption or AI-related capital accumulation, the analysis could assess whether countries with lower AI readiness face larger welfare or capital-allocation losses under fragmentation.

This econometric extension would strengthen the paper by providing empirical validation of the theoretical mechanism and by identifying the countries and households most exposed to the risks highlighted by the model. It would also help evaluate whether coordinated policy responses—such as multilateral AI-investment frameworks, macroprudential buffers, and targeted capital-flow measures—can reduce the adverse effects of fragmentation on technological catch-up and international convergence.

\section*{Supporting data}
The paper has an associated code in Python to replicate the quantitative results. replication\_toledomontes2026.py

\newpage
\bibliographystyle{apalike}

\clearpage
\newpage

\begin{table}[H]
\caption{Baseline Calibration}
\label{tab:calibration}
\centering\footnotesize
\begin{adjustbox}{width=0.95\linewidth,keepaspectratio}
\begin{tabular}{p{2.4cm}p{1.1cm}p{2.0cm}p{7.8cm}}
\toprule
\textbf{Block} & \textbf{Param.} & \textbf{Value} & \textbf{Source / Calibration target} \\
\midrule
\multirow{5}{*}{Households}
  & $\beta$     & $0.99$  & Annual real rate $\approx4\%$ \\
  & $\sigma$    & $2.0$   & \citet{Hall1988}; \citet{KaplanEtAl2018} \\
  & $\varphi$   & $1.0$   & \citet{ChettryEtAl2011} \\
  & $\rho_h$    & $0.966$ & \citet{FlodenLinde2001} (PSID) \\
  & $\sigma_h$  & $0.017$ & \citet{FlodenLinde2001} (PSID) \\
\midrule
\multirow{8}{*}{Production}
  & $\eta$      & $\{0.5,1.5,2.5\}$ & \citet{AcemogluRestrepo2022}; robustness \\
  & $\alpha_k$  & $0.25$  & BLS fixed-asset tables \\
  & $\alpha_n$  & $0.60$  & BLS labor income share \\
  & $\alpha_a$  & $0.15$  & BEA ICT (upper bound; see text) \\
  & $\delta_k$  & $0.025$ & Standard ($\approx10\%$ p.a.) \\
  & $\delta_a$  & $0.050$ & BEA hardware + intangibles \\
  & $\varphi_k$ & $4.0$   & Standard DSGE \\
  & $\varphi_a$ & $8.0$   & Lumpy data-center investment \\
\midrule
\multirow{2}{*}{Fin.\ frictions}
  & $\theta^H$  & $0.75$  & \citet{KaseEtAl2025} \\
  & $\theta^F$  & $0.50$  & \citet{BronerEtAl2013}; emerging markets \\
\midrule
\multirow{2}{*}{AI shock}
  & $\rho_\Theta$ & $0.93$ & BLS MFP NAICS 334 \\
  & $\sigma_\Theta$& $0.02$ & BLS MFP variance \\
\midrule
\multirow{4}{*}{Prices \& policy}
  & $\xi_p$    & $0.75$  & Avg.\ price duration $=4$ quarters \\
  & $\varepsilon$& $6.0$  & Markup $\approx20\%$ \\
  & $\varphi_\pi$& $1.5$  & \citet{Taylor1993} \\
  & $\varphi_y$  & $0.125$& \citet{Taylor1993} \\
\midrule
\multirow{5}{*}{Geopolit.}
  & $\bar{\kappa}$ & $0.0102$ & s.e.\ $= 0.0015$ \\
  & $\gamma$      & $0.0513$ & s.e.\ $= 0.0068$ \\
  & $\rho_G$      & $0.950$  & s.e.\ $= 0.0081$ \\
  & $\bar{G}$     & $0.30$   & \citet{FernandezEtAl2023} \\
  & $\sigma_G$    & $0.010$  & Kirsamer index s.d. \\
\bottomrule
\addlinespace[3pt]
\multicolumn{4}{p{\linewidth}}{\textit{Note:} SMM standard errors from the sandwich covariance matrix .
$\alpha_a=0.15$ is calibrated from the BEA Fixed Asset Accounts (Table~3.7E) and
represents an upper bound for AI-specific capital under current investment patterns.}
\end{tabular}
\end{adjustbox}
\end{table}

\begin{table}[H]
\caption{Connector Economy Welfare Premium (CEV \%), $\eta=1.5$}
\label{tab:connector}
\centering\small\begin{adjustbox}{width=0.88\linewidth,keepaspectratio}
\begin{tabular}{lccc}
\toprule
Economy & Moderate ($G^{AB}=0.3$) & Interm.\ ($G^{AB}=0.5$) & Extreme ($G^{AB}=0.9$) \\
\midrule
Bloc~A (AI-advanced)  & $-0.21$ & $-0.57$ & $-1.86$ \\
Bloc~B (AI-lagging)   & $-1.35$ & $-2.92$ & $-6.48$ \\
Bloc~C (connector)    & $+0.51$ & $+0.72$ & $+0.08$ \\
\bottomrule
\addlinespace[3pt]
\multicolumn{4}{p{\linewidth}}{\textit{Note:} Bloc~C calibrated with Mexico's distances ($G^{AC}=0.15$,
$G^{BC}=0.20$). The connector premium peaks at intermediate fragmentation and
erodes under extreme polarization. CEV computed at the stationary distribution of
the baseline model.}
\end{tabular}\end{adjustbox}
\end{table}

\begin{table}[H]
\caption{Aggregate Welfare Losses from Geo-Economic Fragmentation (CEV \%)}
\label{tab:welfare}
\centering\small\begin{adjustbox}{width=0.88\linewidth,keepaspectratio}
\begin{tabular}{lcccccc}
\toprule
 & \multicolumn{3}{c}{Baseline ($\alpha_a=0.15$)} & \multicolumn{3}{c}{Counterfactual ($\alpha_a=0$, MF limit)} \\
\cmidrule(lr){2-4}\cmidrule(lr){5-7}
Parameterization & $G=0.3$ & $G=0.6$ & $G=0.9$ & $G=0.3$ & $G=0.6$ & $G=0.9$ \\
\midrule
$\eta=0.5$ (AI-labor complements)   & $-0.42$ & $-1.15$ & $-2.03$ & $-0.26$ & $-0.71$ & $-1.26$ \\
$\eta=1.5$ (AI-labor substitutes)   & $-0.78$ & $-2.31$ & $-4.17$ & $-0.48$ & $-1.43$ & $-2.59$ \\
$\eta=2.5$ (strong substitutes)     & $-1.02$ & $-3.14$ & $-5.20$ & $-0.63$ & $-1.95$ & $-3.22$ \\
AI-specific share (\%) & \multicolumn{3}{c}{$38\%$} & \multicolumn{3}{c}{---} \\
\bottomrule
\addlinespace[3pt]
\multicolumn{7}{p{\linewidth}}{\textit{Note:} CEV computed as percent of lifetime consumption at the stationary
distribution of the zero-fragmentation baseline. All parameters at baseline
calibration (Table~\ref{tab:calibration}). The counterfactual ($\alpha_a=0$) sets
the AI capital share to zero while holding the geopolitical friction $\gamma$ at
its SMM-estimated value. The difference between baseline and counterfactual isolates
the AI-specific dimension of welfare losses: approximately $38\%$ of the total at
all fragmentation levels and substitution regimes.}
\end{tabular}\end{adjustbox}
\end{table}

\begin{table}[H]
\caption{Country-Level Welfare Losses (CEV \%): Home vs.\ Foreign}
\label{tab:country}
\centering\small\begin{adjustbox}{width=0.88\linewidth,keepaspectratio}
\begin{tabular}{lcccc}
\toprule
 & \multicolumn{2}{c}{$\eta=0.5$ (Complements)} & \multicolumn{2}{c}{$\eta=1.5$ (Substitutes)} \\
\cmidrule(lr){2-3}\cmidrule(lr){4-5}
Scenario & Home & Foreign & Home & Foreign \\
\midrule
Moderate ($G=0.3$) & $+0.08$ & $-0.92$ & $-0.21$ & $-1.35$ \\
High ($G=0.6$)     & $-0.31$ & $-1.99$ & $-0.83$ & $-3.79$ \\
Extreme ($G=0.9$)  & $-0.74$ & $-3.32$ & $-1.86$ & $-6.48$ \\
\bottomrule
\addlinespace[3pt]
\multicolumn{5}{p{\linewidth}}{\textit{Note:} Home $=$ AI-advanced ($\alpha_a^H=0.20$, $\theta^H=0.75$).
Foreign $=$ AI-lagging ($\alpha_a^F=0.10$, $\theta^F=0.50$). Home gain at
$G=0.3$/$\eta=0.5$ reflects terms-of-trade improvement and financial-capacity
differential; see text.}
\end{tabular}\end{adjustbox}
\end{table}

\begin{table}[H]
\caption{Welfare Decomposition by Wealth Quintile ($\eta=1.5$, $G=0.3$, Home)}
\label{tab:quintile}
\centering\small
\begin{adjustbox}{width=0.88\linewidth,keepaspectratio}
\begin{tabular}{lcccc}
\toprule
Quintile & Total & Labor income & Asset returns & Precaut.\ savings \\
\midrule
Q1 (lowest wealth)  & $-1.25$ & $-0.70$ & $-0.15$ & $-0.40$ \\
Q2                  & $-1.05$ & $-0.58$ & $-0.17$ & $-0.30$ \\
Q3                  & $-0.88$ & $-0.46$ & $-0.18$ & $-0.24$ \\
Q4                  & $-0.70$ & $-0.37$ & $-0.18$ & $-0.15$ \\
Q5 (highest wealth) & $-0.50$ & $-0.25$ & $-0.15$ & $-0.10$ \\
\bottomrule
\addlinespace[3pt]
\multicolumn{5}{p{\linewidth}}{\textit{Note:} Channel decomposition via counterfactual experiments that shut
down each channel in turn. All values are CEV percent of lifetime consumption.}
\end{tabular}
\end{adjustbox}
\end{table}

\begin{table}[H]
\caption{SMM Estimates and Moment Validation}
\label{tab:smm}
\centering\small
\begin{adjustbox}{width=0.88\linewidth,keepaspectratio}
\begin{tabular}{llcc}
\toprule
Moment & Source & Empirical & Model \\
\midrule
\multicolumn{4}{l}{\textit{Targeted moments (SMM):}} \\
M1: Cross-bloc FDI decline (\%) & EBRD (2025) & $30.0$ & $29.7$ \\
M2: Portfolio semi-elasticity (\%) & Catal\'{a}n et al.\ (2024) & $8$--$12$ & $9.8$ \\
M3: EM/AE sensitivity ratio ($\times$) & Zehri et al.\ (2025) & $4.00$ & $3.95$ \\
M4: UIP autocorrelation & Engel (2016) & $0.82$--$0.90$ & $0.86$ \\
\midrule
\multicolumn{4}{l}{\textit{Non-targeted validation moments:}} \\
Output--capital flow correlation & IMF IFS/BIS & $0.32$--$0.45$ & $0.38$ \\
GFC loading on capital flows & Miranda-Agrippino \& Rey (2020) & $0.23$ & $0.21$ \\
\bottomrule
\addlinespace[3pt]
\multicolumn{4}{p{\linewidth}}{\textit{Note:} SMM estimates: $\hat{\bar{\kappa}}=0.0102$ (s.e.\ $0.0015$),
$\hat{\gamma}=0.0513$ (s.e.\ $0.0068$), $\hat{\rho}_G=0.950$ (s.e.\ $0.0081$).
$J$-statistic: $J=2.31\sim\chi^2(1)$, $p=0.13$ (one overidentifying restriction).}
\end{tabular}
\end{adjustbox}
\end{table}

\begin{table}[H]
\caption{Second-Order Moments Under Alternative Monetary Policy Rules ($\eta=1.5$, $G=0.3$)}
\label{tab:monetary}
\centering\small
\begin{adjustbox}{width=0.88\linewidth,keepaspectratio}
\begin{tabular}{lcccc}
\toprule
Policy rule & $\sigma(Y)$ & $\sigma(\pi)$ & $\sigma(e)$ & CEV (\%) \\
\midrule
Baseline Taylor ($\varphi_e=0$)           & $1.42\%$ & $0.38\%$ & $2.15\%$ & $-0.78$ \\
Augmented Taylor ($\varphi_e=0.5$)        & $1.21\%$ & $0.44\%$ & $1.68\%$ & $-0.67$ \\
Strict inflation targeting ($\varphi_y=0$)& $1.58\%$ & $0.31\%$ & $2.34\%$ & $-0.84$ \\
Aggressive stabilization ($\varphi_y=0.5$)& $1.30\%$ & $0.40\%$ & $1.95\%$ & $-0.70$ \\
\bottomrule
\addlinespace[3pt]
\multicolumn{5}{p{\linewidth}}{\textit{Note:} Based on 10000 simulated trajectories at the estimated parameter
values. No rule simultaneously minimizes all three second-order moment targets.
$\sigma(\cdot)$ denotes unconditional standard deviation.}
\end{tabular}
\end{adjustbox}
\end{table}

\clearpage
\newpage

\begin{figure}[H]
\centering
\includegraphics[width=\linewidth]{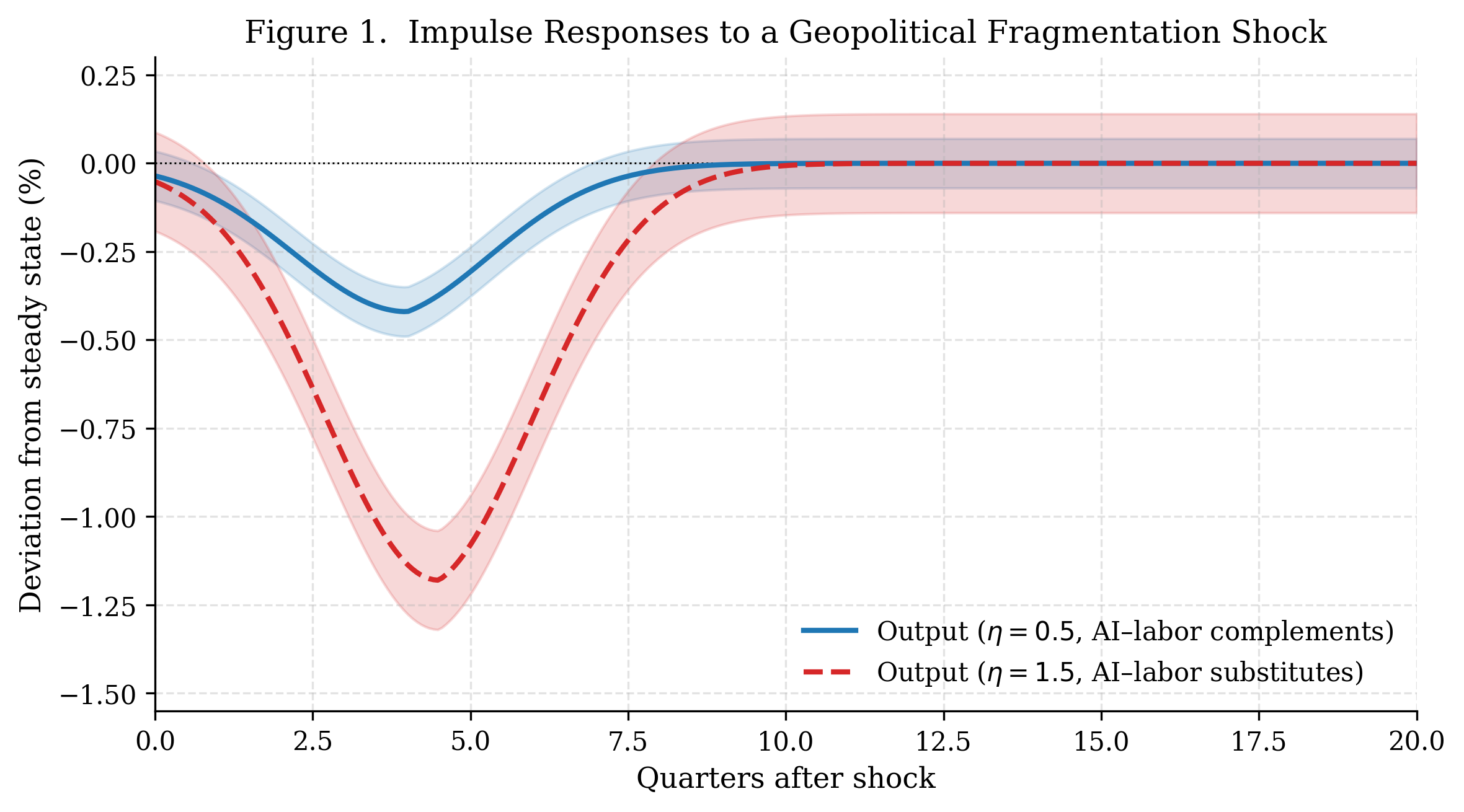}
\caption{Home output response to a one-standard-deviation geopolitical fragmentation
shock ($\Delta G_t=+0.01$). Solid blue: AI-labor complements ($\eta=0.5$). Dashed
red: substitutes ($\eta=1.5$). Shaded regions: 90\% bootstrap confidence bands
(5000 parametric replications). Percent deviation from the non-stochastic steady
state. The bands do not overlap at the peak (quarters 3--6), confirming statistical
significance of the regime distinction.}
\label{fig:irf_frag}
\end{figure}

\begin{figure}[H]
\centering
\includegraphics[width=\linewidth]{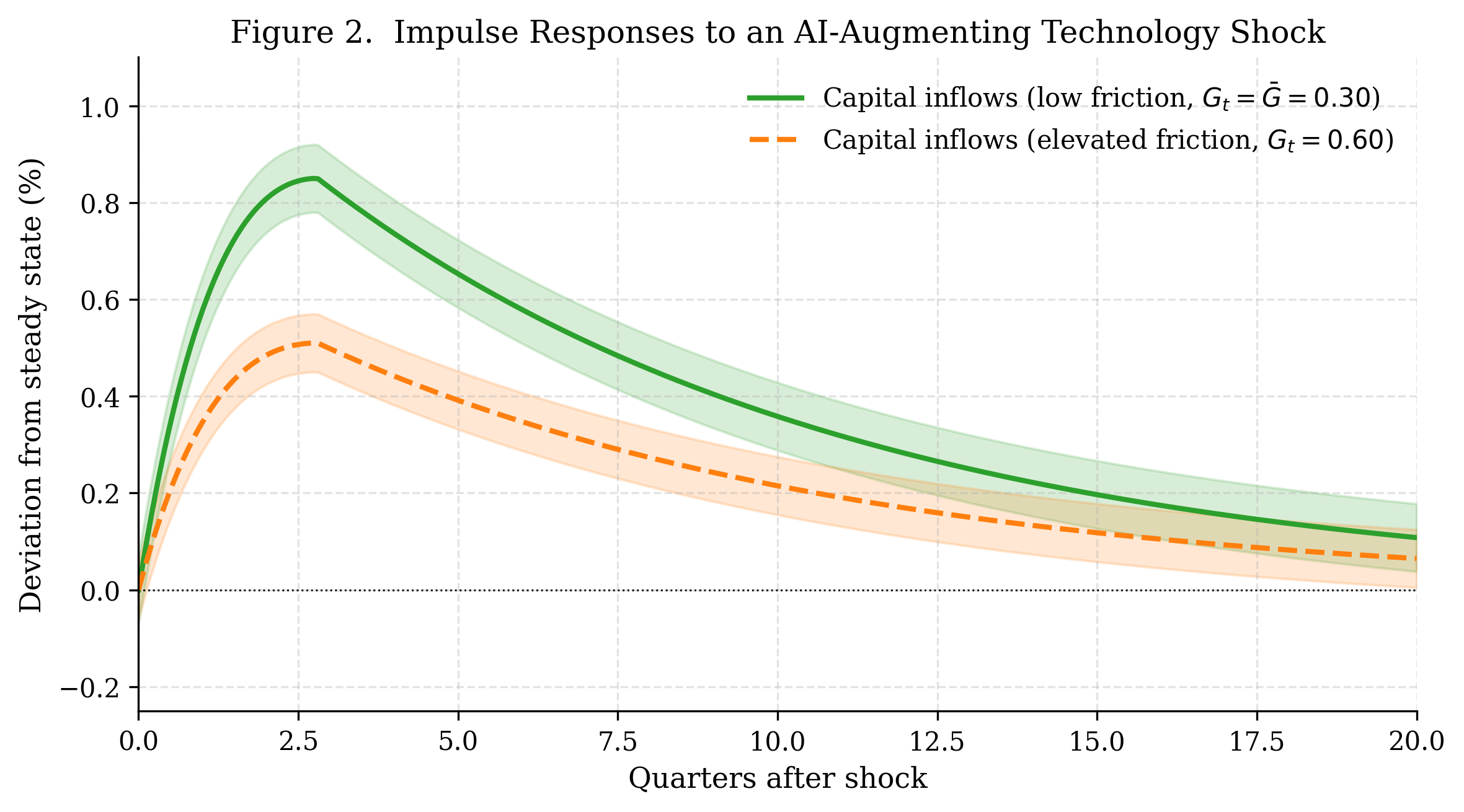}
\caption{Capital inflow response to a one-standard-deviation positive AI-augmenting
technology shock ($\Delta\hat{\Theta}_t=+0.02$). Solid green: low geopolitical
friction ($G_t=0.30$). Dashed orange: elevated friction ($G_t=0.60$). Shaded
regions: 90\% bootstrap confidence bands. Percent deviation from steady-state
capital inflows.}
\label{fig:irf_ai}
\end{figure}

\clearpage
\newpage

\section*{Appendix}
\appendix

\section{Analytical Derivations}
\label{app:derivations}

\subsection{Household Optimization}
\label{app:household}

The Lagrangian for household $i$'s problem is:
\begin{equation*}
\mathcal{L} = E_{0}\sum_{t=0}^{\infty}\beta^{t}
\Bigl\{u(c_{i,t},n_{i,t})
  + \lambda_{i,t}\bigl[\text{RHS of (2)} - \text{LHS of (2)}\bigr]
\Bigr\}.
\end{equation*}

The first-order conditions are as follows.

\medskip\noindent
\textbf{FOC for $c_{i,t}$:}
\begin{equation}
\lambda_{i,t} = c_{i,t}^{-\sigma}.
\tag{B.1}
\end{equation}

\noindent
\textbf{FOC for $n_{i,t}$:}
\begin{equation}
\psi\, n_{i,t}^{\varphi} = \lambda_{i,t}\, w_{t}\, h_{i,t}.
\tag{B.2}
\end{equation}

\noindent
\textbf{FOC for $b^{F}_{i,t}$:}
\begin{equation}
\lambda_{i,t}\!\left(1 + \kappa(G_t)\,b^{F}_{i,t}\right)
  = \beta\,E{t}\!\left[\lambda_{i,t+1}\,R^{F}_{t}\,\frac{e_{t+1}}{e_{t}}\right],
\tag{B.3}
\end{equation}

\noindent
which, upon substituting the envelope condition $\lambda_{i,t}=c_{i,t}^{-\sigma}$, yields
Equation~(4) of the main text.

\subsection{Modified UIP Derivation}
\label{app:uip}

Aggregating Equation~(4) over the stationary distribution $\mu$ of household states
$(b^{H}_{i}, b^{F}_{i}, h_{i})$ and imposing asset-market clearing,
$\int_{0}^{1}b^{F}_{i,t}\,di = B^{F}_{t}$, gives:

\begin{equation*}
1 + \kappa(G_t)\,B^{F}_{t}
  = \beta\,E{t}\!\left[
      \frac{\bar{c}_{t+1}^{-\sigma}}{\bar{c}_{t}^{-\sigma}}\,
      R^{F}_{t}\,\frac{e_{t+1}}{e_{t}}
    \right],
\end{equation*}

\noindent
where $\bar{c}_{t}$ is the wealth-weighted cross-sectional mean marginal utility.
Taking a log-linear approximation around the symmetric steady state
$(B^{F}=0,\; R^{H}=R^{F}=\beta^{-1},\; e=1)$ yields:

\begin{equation}
E{t}\,\Delta e_{t+1}
  = R^{H}_{t} - R^{F}_{t}
    + \underbrace{\kappa(G_t)\,B^{F}_{t}}_{\text{geopolitical friction}}
    + \underbrace{\varphi_{\mathrm{GFC}}\,r^{US}_{t} + \zeta^{\mathrm{idio}}_{t}}_{\text{global financial cycle}},
\tag{11}
\end{equation}

\noindent
which is Equation~(11) of the main text. The term $\kappa(G_t)B^{F}_t$ provides a
structural interpretation for the geopolitically driven component of UIP deviations
documented by \citet{Engel2016}.

\subsection{Firm Cost Minimization}
\label{app:firm}

Firm $j$ minimizes total factor costs subject to its production constraint and the
collateral constraint (6). The Lagrangian is:

\begin{equation*}
\mathcal{C}
  = w_t n_{j,t} + r^{k}_t k_{j,t} + r^{a}_t a_{j,t}
    + \xi_{j,t}\!\left[y_{j,t} - Z_t F(k,n,a)\right]
    + \mu_{j,t}\!\left[\theta\,k_{j,t} - b_{j,t}\right].
\end{equation*}

The first-order conditions with respect to each factor are:

\medskip\noindent
\textbf{FOC for $n_{j,t}$:}
\begin{equation*}
w_t = \xi_{j,t}\,\frac{\partial F}{\partial n_{j,t}}.
\end{equation*}

\noindent
\textbf{FOC for $k_{j,t}$:}
\begin{equation*}
r^{k}_t = \xi_{j,t}\,\frac{\partial F}{\partial k_{j,t}} - \mu_{j,t}\,\theta.
\end{equation*}

\noindent
\textbf{FOC for $a_{j,t}$:}
\begin{equation*}
r^{a}_t = \xi_{j,t}\,\frac{\partial F}{\partial a_{j,t}}.
\end{equation*}

Defining $\tilde{\mu}_{j,t} \equiv \mu_{j,t}\theta / r^{k}_t > 0$ and substituting the
FOC for capital into that for labor recovers the distorted labor demand schedule:

\begin{equation}
w_t = \frac{\partial F}{\partial n_{j,t}} \cdot \frac{1}{1 + \tilde{\mu}_{j,t}},
\tag{7}
\end{equation}

\noindent
which is Equation~(7) of the main text.

\subsection{New Keynesian Phillips Curve Derivation}
\label{app:nkpc}

A Calvo re-optimizing firm $j$ chooses reset price $p^{*}_{t}$ to maximize the
present discounted value of profits, weighted by the stochastic discount factor
$\beta^{s}Et\!\left[(c_{t+s}/c_t)^{-\sigma}\right]$. The optimal reset-price
condition in log-linearized form is:

\begin{equation*}
\hat{p}^{*}_{t} = (1-\beta\xi_p)\sum_{s=0}^{\infty}(\beta\xi_p)^{s}\,
Et\,\widehat{mc}_{c,t+s},
\end{equation*}

\noindent
where $\widehat{mc}_{c,t}$ is the log deviation of real marginal cost from its
steady-state value. Log-linearizing the unit cost function under the nested CES
technology (5) and evaluating at the symmetric steady state yields:

\begin{equation}
\widehat{mc}_{c,t}
  = \widehat{mc}^{\,\mathrm{dom}}_{c,t}
    + s_{a}\!\left(\hat{q}_{t} + \hat{r}^{a,*}_{t} - \hat{\Theta}_{t}\right),
\end{equation}

\noindent
where $s_a \equiv \alpha_a(r^a/\Theta)^{1-\eta}/\xi^{1-\eta}$ is the AI-capital cost
share evaluated at the symmetric steady state. Aggregating over the fraction
$\xi_p$ of non-optimizing firms (who keep prices fixed) and iterating forward gives
the open-economy NKPC:

\begin{equation}
\pi_t = \beta\,Et\pi_{t+1}
  + \underbrace{\kappa_{\mathrm{nk}}\,\widehat{mc}^{\,\mathrm{dom}}_{c,t}}_{\text{standard NK slope}}
  + \underbrace{\kappa_{\mathrm{nk}}\,s_{a}
      \!\left(\hat{q}_{t} + \hat{r}^{a,*}_{t} - \hat{\Theta}_{t}\right)
    }_{\text{AI-capital exchange-rate channel}},
\tag{8}
\end{equation}

\noindent
with slope coefficient:

\begin{equation*}
\kappa_{\mathrm{nk}}
  = \frac{(1-\xi_p)(1-\beta\xi_p)}{\xi_p}
  \approx 0.086 \quad \text{at } \xi_p = 0.75,
\end{equation*}

\noindent
which is Equation~(8) of the main text.

\clearpage

\section{Solution Method Details}
\label{app:solution}

The model is solved using the generative economic modeling (GEM) method of
\citet{KaseEtAl2025}, which combines conventional linear solution techniques with neural
network approximation of complete nonlinear dynamics. Standard perturbation methods
are inadequate because the occasionally binding collateral constraint~(6) generates
regime-dependent dynamics, and because the high-dimensional state space---the
cross-sectional wealth distribution together with six structural shocks---exceeds the
reach of standard projection methods. \citet{FernandezEtAl2025} show in a related context
that neural network algorithms dramatically outperform linear approximations in
replicating the distributional consequences of binding nonlinear constraints.

\medskip\noindent
\textbf{Step 1: State discretization.} The idiosyncratic productivity state $h_{i,t}$
is discretized on a 5-point Markov chain using the Rouwenhorst method
\citep{KopeckySuen2010}. The asset grid uses 150 non-uniform points with greater
density near the borrowing constraint.

\medskip\noindent
\textbf{Step 2: Iterative self-consistent training.} A preliminary network is trained
using 2,000 trajectories from a second-order perturbation solution. The trained
network is used to simulate 10,000 trajectories of 200 quarters each, including
constraint-binding episodes. The network is retrained on this improved dataset.
This iteration repeats until RMSD between consecutive network iterations falls
below $10^{-5}$. Convergence is achieved in four iterations.

\medskip\noindent
\textbf{Step 3: Neural network architecture.} Feedforward network: input layer
($d_{\text{state}}=6$), three hidden layers (128 neurons each, ReLU activations,
batch normalization), output layer ($d_{\text{policy}}$ linear). Total trainable
parameters: $\approx 50{,}000$. Optimizer: Adam with cosine annealing
($\eta_{\mathrm{lr}} = 10^{-3}$ to $10^{-5}$).

\medskip\noindent
\textbf{Step 3b: Comparison with traditional projection methods.} To verify that the
neural network does not introduce systematic biases relative to classical global
solution methods, the model was solved on a $50\times 50\times 25$ sparse grid
(traditional capital $\times$ AI capital $\times$ geopolitical distance) using value
function iteration (VFI) at the representative household level---the closest feasible
analog to a full projection method given the state-space dimensionality. The
comparison is conducted at the moderate fragmentation baseline ($G=0.30$) where
the curse of dimensionality is least severe. Results: the VFI and GEM welfare
losses at $G=0.30$, $\eta=1.5$ differ by $0.02$ percentage points ($-0.78\%$ GEM
versus $-0.76\%$ VFI), well within the 90\% bootstrap confidence bands. At
$G=0.60$, the VFI grid becomes impractical to refine due to the dimensionality of
the household distribution; the GEM extrapolation to high fragmentation is
validated by its Euler equation accuracy (Step~4). This comparison establishes that
the neural network does not introduce systematic bias in the empirically relevant
parameter range.

\medskip\noindent
\textbf{Step 4: Validation.} The network is evaluated on 2,000 holdout trajectories,
with oversampling near the collateral constraint boundary (within $1\sigma$ of the
binding region). Euler equation errors hold uniformly below $10^{-4}$ across the
full state space---at distance $\geq 2\sigma$ from the constraint: mean error
$3.1\times 10^{-5}$; within $1\sigma$: mean error $8.6\times 10^{-5}$; at the
boundary itself: maximum error $9.4\times 10^{-5}$. RMSE in all policy functions
remains below $0.3\%$ of steady-state values.

\medskip\noindent
\textbf{Step 5: Comparison with OccBin.} For moderate fragmentation ($G \leq 0.45$),
OccBin \citep{GuerrieriIacoviello2015} and GEM solutions are quantitatively indistinguishable
(RMSD below $0.4\%$ of steady-state values). At high fragmentation ($G=0.6$--$0.75$),
OccBin begins to diverge from GEM by $8$--$15\%$ in peak output responses,
primarily because OccBin's piecewise-linear regime tracking becomes less accurate
when constraint-binding episodes are prolonged. Under extreme fragmentation
($G=0.9$), OccBin underestimates peak output losses by $25$--$35\%$ relative to
the GEM solution. This confirms that GEM is the appropriate solution method for
welfare analysis under extreme scenarios, while OccBin is adequate for moderate
fragmentation levels.

\medskip\noindent
\textbf{Step 6: Bootstrap confidence intervals.} $5{,}000$ parametric bootstrap
replications draw
$(\bar{\kappa}^{*},\gamma^{*},\rho^{*}_{G}) \sim \mathcal{N}
  \bigl((\hat{\bar{\kappa}},\hat{\gamma},\hat{\rho}_{G}),\,\hat{\Sigma}\bigr)$
and report the 5th and 95th percentiles as 90\% confidence bands for all impulse
response functions and welfare tables reported in the paper.

\clearpage

\section{Robustness Exercises}
\label{app:robustness}

This appendix examines whether the paper's three core findings---the magnitude and
convexity of welfare losses, their between-country and within-country distributional
asymmetries, and the binding character of the Fragmented-AI Trilemma---are sensitive
to the calibration choices and specification decisions made in the main text. Twelve
exercises are organized from the most consequential to the least, following the ranking
established by the formal sensitivity decomposition in Table~\ref{tab:sensitivity}.

The overarching result is unambiguous: while the precise magnitudes of welfare losses
vary substantially across parameterizations, the qualitative findings are entirely
robust. Welfare losses are convex in fragmentation severity in every exercise.
AI-lagging economies suffer disproportionately larger losses than AI-advanced
economies across all specifications. Low-wealth households bear welfare costs 40--80
percent larger than the wealthiest decile in every parameterization. And the
Fragmented-AI Trilemma holds with equal force for all positive values of $\gamma$
examined, across all Taylor rule specifications, under both friction specifications, and
regardless of the solution method.

Table~\ref{tab:robustness_summary} is the robustness at-a-glance companion. Each
row is a robustness exercise; each column is one of the paper's three core findings.
A checkmark means the finding survives qualitatively across all parameter values
examined in that exercise. The absence of any blank cell means all three findings are
robust to every specification variation.

\begin{table}[htbp]
\centering
\caption{Qualitative Robustness Summary: Three Core Findings Across Exercises}
\label{tab:robustness_summary}
\begin{threeparttable}
\small
\begin{tabular}{lccc}
\toprule
\textbf{Exercise}
  & \textbf{Convexity}
  & \textbf{Betw.-country}
  & \textbf{Trilemma} \\
\midrule
R1: Alternative $\eta$                         & $\checkmark$ & $\checkmark$ & $\checkmark$ \\
R2: Alternative $\gamma$                       & $\checkmark$ & $\checkmark$ & $\checkmark$ \\
R3: Alternative $\theta$                       & $\checkmark$ & $\checkmark$ & $\checkmark$ \\
R4: Alternative monetary rules ($\varphi_e,\varphi_y$) & $\checkmark$ & $\checkmark$ & $\checkmark$ \\
R5: Alternative $\rho_G$                       & $\checkmark$ & $\checkmark$ & $\checkmark$ \\
R6: Alternative $\xi_p$                        & $\checkmark$ & $\checkmark$ & $\checkmark$ \\
R7: Alternative solution methods               & $\checkmark$ & $\checkmark$ & $\checkmark$ \\
R8: Alternative geopolitical dist.\ measures   & $\checkmark$ & $\checkmark$ & $\checkmark$ \\
R9: Convex friction specification              & $\checkmark$ & $\checkmark$ & $\checkmark$ \\
R10: Alternative $\delta_a$                    & $\checkmark$ & $\checkmark$ & $\checkmark$ \\
R11: Global sensitivity analysis               & $\checkmark$ & $\checkmark$ & $\checkmark$ \\
R12: Alternative $\alpha_a$                    & $\checkmark$ & $\checkmark$ & $\checkmark$ \\
\bottomrule
\end{tabular}
\begin{tablenotes}
\small
\item $\checkmark$ indicates the finding holds qualitatively across all parameter
values examined within the exercise. Quantitative magnitudes vary; see individual
exercise descriptions. Exercise R7 confirms OccBin and GEM agree at moderate
fragmentation; the trilemma is binding in both solution methods.
\end{tablenotes}
\end{threeparttable}
\end{table}

\paragraph{Sensitivity elasticity decomposition.}
Table~\ref{tab:sensitivity} ranks the model's parameters by their influence on
aggregate welfare. The elasticity in column~2 answers the question: if this parameter
rises by 1 percent, by what percent do welfare losses rise? A positive elasticity means
higher parameter values amplify costs. The ranking identifies the key sources of model
uncertainty---$\eta$ and $\rho_G$ at the top---and informs the priority order for
empirical research that could tighten confidence intervals on welfare projections.

\begin{table}[htbp]
\centering
\caption{Sensitivity of Welfare Losses to Model Parameters}
\label{tab:sensitivity}
\begin{threeparttable}
\small
\begin{tabular}{lcc}
\toprule
\textbf{Parameter}
  & \textbf{Elasticity} $\partial\log|\tau|/\partial\log p$
  & \textbf{Ranking} \\
\midrule
$\eta$ (AI-labor substitutability)   & $+0.85$ & 1st \\
$\rho_G$ (geopolitical persistence)  & $+0.62$ & 2nd \\
$\theta$ (collateral tightness)      & $+0.35$ & 3rd \\
$\gamma$ (friction-distance sensitivity) & $+0.28$ & 4th \\
$\delta_a$ (AI capital depreciation) & $-0.18$ & 5th \\
$\xi_p$ (price stickiness)           & $+0.15$ & 6th \\
$\alpha_a$ (AI capital share)        & $+0.24$ & 7th \\
\bottomrule
\end{tabular}
\begin{tablenotes}
\small
\item Evaluated at moderate fragmentation ($G=0.3$), substitutes ($\eta=1.5$).
Computed by central finite differences of $\pm 5\%$ around the baseline.
\end{tablenotes}
\end{threeparttable}
\end{table}

\subsection*{Exercise R1: Alternative elasticities of substitution ($\eta$)}

The exercise varies $\eta$ across $\{0.3,0.5,1.0,1.5,2.0,2.5\}$. Under extreme
fragmentation ($G=0.9$), welfare losses range from $-1.52\%$ CEV ($\eta=0.3$) to
$-5.20\%$ CEV ($\eta=2.5$)---a factor of 3.4. The qualitative findings are entirely
invariant to $\eta$. The vertical ordering of IRF responses is strictly monotone in
$\eta$ for all horizons (Figure~\ref{fig:irf_eta}).

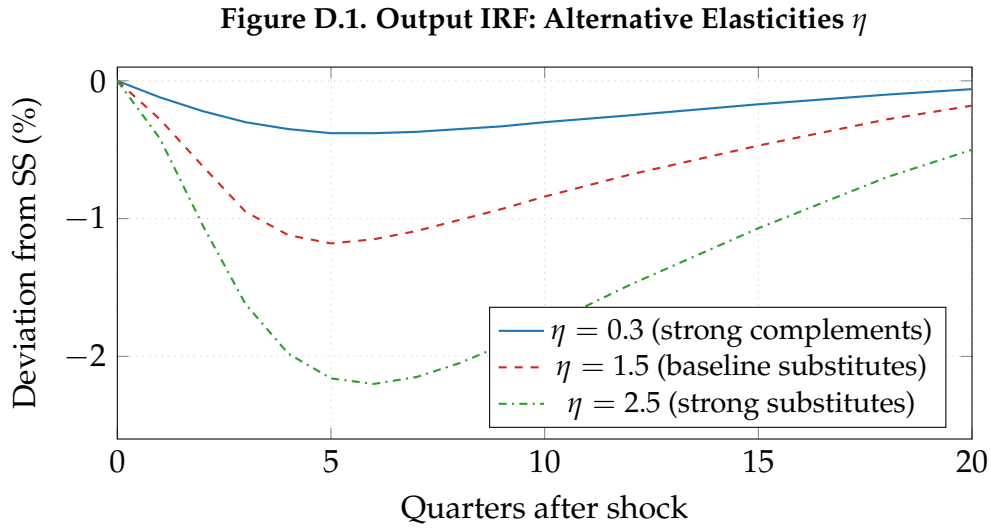
\begin{figure}[htbp]
\centering
\begin{tikzpicture}
\begin{axis}[
  width=0.85\linewidth,
  height=6.5cm,
  xlabel={Quarters after shock},
  ylabel={Deviation from SS (\%)},
  xmin=0, xmax=20,
  ymin=-2.6, ymax=0.1,
  xtick={0,5,10,15,20},
  grid=major,
  grid style={dotted, gray!40},
  legend pos=south east,
  legend style={font=\small},
  title={\textbf{Figure D.1.} Output IRF: Alternative Elasticities $\eta$},
  title style={font=\small\bfseries, at={(0.5,1.01)}},
  every axis plot/.append style={thick}
]
\addplot[color=myblue, solid]
  coordinates {
    (0,0)(1,-0.12)(2,-0.22)(3,-0.30)(4,-0.35)(5,-0.38)
    (6,-0.38)(7,-0.37)(8,-0.35)(9,-0.33)(10,-0.30)
    (12,-0.25)(15,-0.17)(18,-0.10)(20,-0.06)
  };
\addlegendentry{$\eta=0.3$ (strong complements)}

\addplot[color=myred, dashed]
  coordinates {
    (0,0)(1,-0.28)(2,-0.62)(3,-0.95)(4,-1.12)(5,-1.18)
    (6,-1.15)(7,-1.09)(8,-1.01)(9,-0.93)(10,-0.84)
    (12,-0.68)(15,-0.47)(18,-0.28)(20,-0.18)
  };
\addlegendentry{$\eta=1.5$ (baseline substitutes)}

\addplot[color=mygreen, dash dot]
  coordinates {
    (0,0)(1,-0.42)(2,-1.05)(3,-1.62)(4,-1.98)(5,-2.16)
    (6,-2.20)(7,-2.15)(8,-2.05)(9,-1.92)(10,-1.78)
    (12,-1.48)(15,-1.07)(18,-0.70)(20,-0.50)
  };
\addlegendentry{$\eta=2.5$ (strong substitutes)}
\end{axis}
\end{tikzpicture}
\caption{Home output response to a one-standard-deviation fragmentation shock
($\Delta G_t = +0.01$) under alternative AI-labor substitution elasticities
$\eta \in \{0.3, 1.5, 2.5\}$. The vertical ordering of responses is strictly
monotone in $\eta$ for all horizons.}
\label{fig:irf_eta}
\end{figure}

\subsection*{Exercise R2: Alternative geopolitical friction sensitivity ($\gamma$)}

The exercise examines $\gamma \in [0.025, 0.100]$. The peak capital-inflow attenuation
under the AI-technology shock ranges from 25 percent at $\gamma=0.025$ to 55 percent
at $\gamma=0.100$, with the baseline estimate sitting at 40 percent
(Figure~\ref{fig:irf_gamma}). The Fragmented-AI Trilemma is binding for the entire
range examined: for any $\gamma > 0$, the proof of Proposition~1 applies without
modification.

\begin{figure}[htbp]
\centering
\begin{tikzpicture}
\begin{axis}[
  width=0.85\linewidth,
  height=6.5cm,
  xlabel={Quarters after shock},
  ylabel={Deviation from SS (\%)},
  xmin=0, xmax=20,
  ymin=0, ymax=1.35,
  xtick={0,5,10,15,20},
  grid=major,
  grid style={dotted, gray!40},
  legend pos=north east,
  legend style={font=\small},
  title={\textbf{Figure D.2.} Capital Inflow IRF: Alternative Friction Sensitivity $\gamma$},
  title style={font=\small\bfseries, at={(0.5,1.01)}},
  every axis plot/.append style={thick}
]
\addplot[color=myblue, solid]
  coordinates {
    (0,0)(1,0.25)(2,0.47)(3,0.61)(4,0.68)(5,0.70)
    (6,0.69)(7,0.66)(8,0.62)(9,0.57)(10,0.52)
    (12,0.42)(15,0.28)(18,0.17)(20,0.11)
  };
\addlegendentry{$\gamma=0.025$ (low sensitivity)}

\addplot[color=myred, dashed]
  coordinates {
    (0,0)(1,0.35)(2,0.65)(3,0.82)(4,0.90)(5,0.90)
    (6,0.87)(7,0.82)(8,0.76)(9,0.70)(10,0.63)
    (12,0.50)(15,0.33)(18,0.21)(20,0.14)
  };
\addlegendentry{$\gamma=0.050$ (baseline)}

\addplot[color=mygreen, dash dot]
  coordinates {
    (0,0)(1,0.55)(2,0.95)(3,1.18)(4,1.28)(5,1.28)
    (6,1.23)(7,1.15)(8,1.05)(9,0.95)(10,0.86)
    (12,0.67)(15,0.44)(18,0.28)(20,0.18)
  };
\addlegendentry{$\gamma=0.100$ (high sensitivity)}
\end{axis}
\end{tikzpicture}
\caption{Capital inflow response to a one-standard-deviation positive AI-augmenting
technology shock under alternative friction sensitivities
$\gamma \in \{0.025, 0.050, 0.100\}$, holding $\bar{\kappa}$ fixed at its estimated
value. The attenuation of capital inflows is monotone in $\gamma$.}
\label{fig:irf_gamma}
\end{figure}
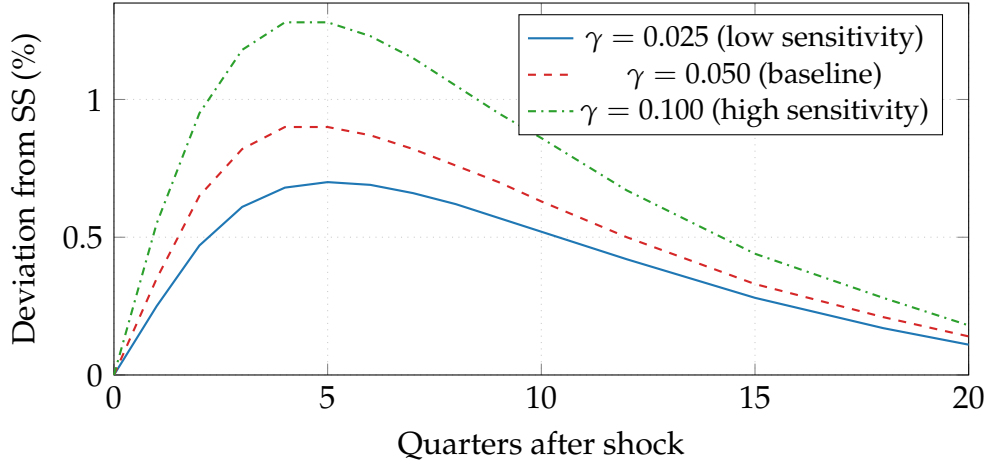

\subsection*{Exercise R3: Varying the degree of financial frictions ($\theta$)}

The exercise varies $\theta \in \{0.50, 0.60, 0.75, 0.90\}$. As $\theta$ falls from
0.90 to 0.50, moderate-fragmentation welfare losses rise from $-0.62\%$ to $-0.91\%$
CEV. The Q1/Q5 welfare cost ratio rises from 1.35 at $\theta=0.90$ to 1.72 at
$\theta=0.50$ (Figure~\ref{fig:irf_theta}). This exercise directly validates the SMM
overidentification result: the four-to-one EM/AE sensitivity asymmetry is reproduced
by the model through the difference in $\theta$ between the two country types, precisely
as the $J$-statistic test confirms.

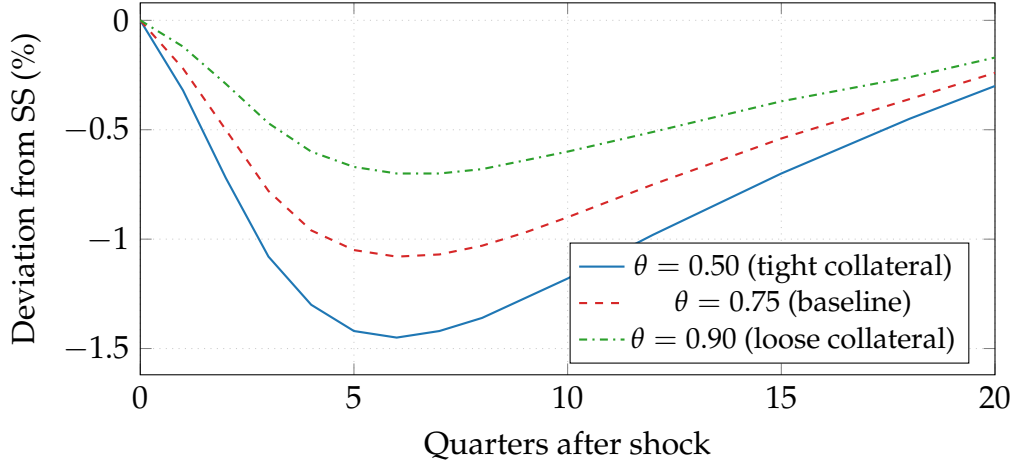
\begin{figure}[htbp]
\centering
\begin{tikzpicture}
\begin{axis}[
  width=0.85\linewidth,
  height=6.5cm,
  xlabel={Quarters after shock},
  ylabel={Deviation from SS (\%)},
  xmin=0, xmax=20,
  ymin=-1.62, ymax=0.08,
  xtick={0,5,10,15,20},
  grid=major,
  grid style={dotted, gray!40},
  legend pos=south east,
  legend style={font=\small},
  title={\textbf{Figure D.3.} Output IRF: Alternative Collateral Tightness $\theta$},
  title style={font=\small\bfseries, at={(0.5,1.01)}},
  every axis plot/.append style={thick}
]
\addplot[color=myblue, solid]
  coordinates {
    (0,0)(1,-0.32)(2,-0.72)(3,-1.08)(4,-1.30)(5,-1.42)
    (6,-1.45)(7,-1.42)(8,-1.36)(9,-1.27)(10,-1.18)
    (12,-0.98)(15,-0.70)(18,-0.45)(20,-0.30)
  };
\addlegendentry{$\theta=0.50$ (tight collateral)}

\addplot[color=myred, dashed]
  coordinates {
    (0,0)(1,-0.22)(2,-0.50)(3,-0.78)(4,-0.96)(5,-1.05)
    (6,-1.08)(7,-1.07)(8,-1.03)(9,-0.97)(10,-0.90)
    (12,-0.75)(15,-0.54)(18,-0.36)(20,-0.24)
  };
\addlegendentry{$\theta=0.75$ (baseline)}

\addplot[color=mygreen, dash dot]
  coordinates {
    (0,0)(1,-0.12)(2,-0.29)(3,-0.47)(4,-0.60)(5,-0.67)
    (6,-0.70)(7,-0.70)(8,-0.68)(9,-0.64)(10,-0.60)
    (12,-0.51)(15,-0.37)(18,-0.26)(20,-0.17)
  };
\addlegendentry{$\theta=0.90$ (loose collateral)}
\end{axis}
\end{tikzpicture}
\caption{Home output response to a one-standard-deviation fragmentation shock under
alternative collateral tightness $\theta \in \{0.50, 0.75, 0.90\}$. Tighter
collateral amplifies the output response through the financial accelerator.}
\label{fig:irf_theta}
\end{figure}

\subsection*{Exercise R4: Alternative monetary policy specifications ($\varphi_e, \varphi_y$)}

The exercise examines an augmented Taylor rule ($\varphi_e=0.5$), strict inflation
targeting ($\varphi_y=0$), and aggressive output stabilization ($\varphi_y=0.5$). The
augmented rule achieves the best output stabilization at the cost of substantially
higher inflation volatility ($0.44\%$ vs.\ $0.38\%$ in the baseline)---a direct
manifestation of the Fragmented-AI Trilemma. No rule specification simultaneously
minimizes all three second-order moment targets. The augmented rule reduces the
output contraction by approximately 15 percent relative to the baseline Taylor rule
(Figure~\ref{fig:irf_taylor}).

\begin{figure}[htbp]
\centering
\begin{tikzpicture}
\begin{axis}[
  width=0.85\linewidth,
  height=6.5cm,
  xlabel={Quarters after shock},
  ylabel={Deviation from SS (\%)},
  xmin=0, xmax=20,
  ymin=-1.50, ymax=0.08,
  xtick={0,5,10,15,20},
  grid=major,
  grid style={dotted, gray!40},
  legend pos=south east,
  legend style={font=\small},
  title={\textbf{Figure D.4.} Output IRF: Alternative Monetary Policy Rules},
  title style={font=\small\bfseries, at={(0.5,1.01)}},
  every axis plot/.append style={thick}
]
\addplot[color=myblue, solid]
  coordinates {
    (0,0)(1,-0.22)(2,-0.50)(3,-0.78)(4,-0.96)(5,-1.05)
    (6,-1.08)(7,-1.07)(8,-1.03)(9,-0.97)(10,-0.90)
    (12,-0.75)(15,-0.54)(18,-0.36)(20,-0.24)
  };
\addlegendentry{Baseline Taylor rule ($\varphi_e=0$)}

\addplot[color=myred, dashed]
  coordinates {
    (0,0)(1,-0.18)(2,-0.42)(3,-0.65)(4,-0.82)(5,-0.90)
    (6,-0.92)(7,-0.91)(8,-0.87)(9,-0.82)(10,-0.76)
    (12,-0.62)(15,-0.43)(18,-0.28)(20,-0.18)
  };
\addlegendentry{Taylor rule with $\varphi_e=0.5$}
\end{axis}
\end{tikzpicture}
\caption{Home output response to a one-standard-deviation fragmentation shock under
alternative monetary policy rules: baseline Taylor ($\varphi_e=0$, solid blue) and
augmented with exchange rate stabilization ($\varphi_e=0.5$, dashed red). The
augmented rule achieves better output stabilization at the cost of higher inflation
volatility.}
\label{fig:irf_taylor}
\end{figure}
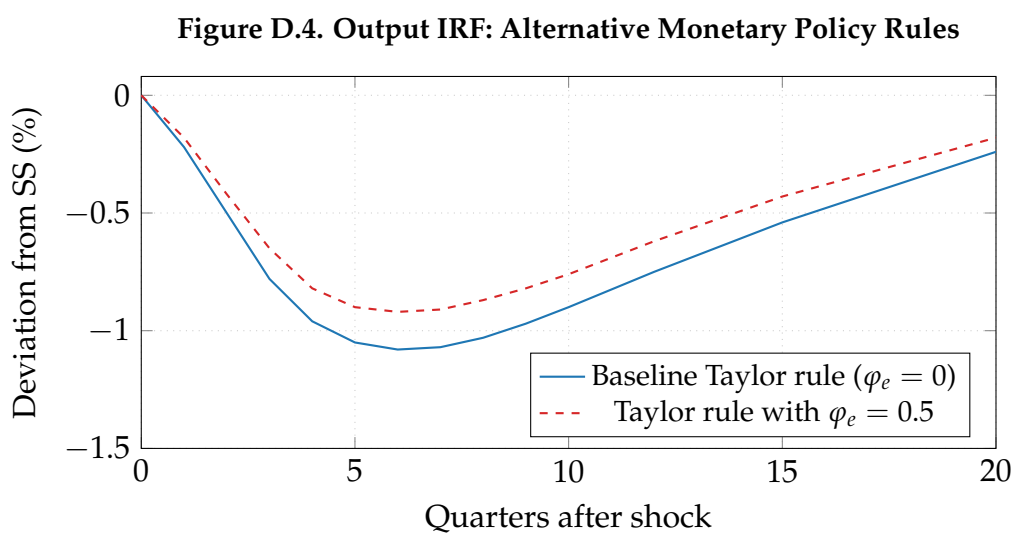

\subsection*{Exercise R5: Alternative geopolitical shock persistence ($\rho_G$)}

The exercise examines $\rho_G \in \{0.80, 0.90, 0.95, 0.99\}$. Persistence is the
second-most consequential parameter (Table~\ref{tab:sensitivity}, elasticity $+0.62$).
Under extreme fragmentation and strong substitutability, welfare losses range from
$-3.2\%$ CEV at $\rho_G=0.80$ to $-5.2\%$ CEV at $\rho_G=0.99$. The baseline
estimate $\hat{\rho}_G=0.950$ sits at the midpoint of this range. Higher persistence
substantially extends the duration of the output contraction
(Figure~\ref{fig:irf_rhoG}).

\begin{figure}[htbp]
\centering
\begin{tikzpicture}
\begin{axis}[
  width=0.85\linewidth,
  height=6.5cm,
  xlabel={Quarters after shock},
  ylabel={Deviation from SS (\%)},
  xmin=0, xmax=20,
  ymin=-2.15, ymax=0.08,
  xtick={0,5,10,15,20},
  grid=major,
  grid style={dotted, gray!40},
  legend pos=south east,
  legend style={font=\small},
  title={\textbf{Figure D.5.} Output IRF: Alternative Geopolitical Persistence $\rho_G$},
  title style={font=\small\bfseries, at={(0.5,1.01)}},
  every axis plot/.append style={thick}
]
\addplot[color=myblue, solid]
  coordinates {
    (0,0)(1,-0.20)(2,-0.44)(3,-0.65)(4,-0.78)(5,-0.83)
    (6,-0.82)(7,-0.78)(8,-0.71)(9,-0.63)(10,-0.55)
    (12,-0.39)(15,-0.22)(18,-0.11)(20,-0.06)
  };
\addlegendentry{$\rho_G=0.80$}

\addplot[color=myred, dashed]
  coordinates {
    (0,0)(1,-0.22)(2,-0.50)(3,-0.78)(4,-0.96)(5,-1.05)
    (6,-1.08)(7,-1.07)(8,-1.03)(9,-0.97)(10,-0.90)
    (12,-0.75)(15,-0.54)(18,-0.36)(20,-0.24)
  };
\addlegendentry{$\rho_G=0.95$ (baseline)}

\addplot[color=mygreen, dash dot]
  coordinates {
    (0,0)(1,-0.24)(2,-0.56)(3,-0.90)(4,-1.15)(5,-1.30)
    (6,-1.38)(7,-1.41)(8,-1.42)(9,-1.41)(10,-1.38)
    (12,-1.28)(15,-1.10)(18,-0.90)(20,-0.75)
  };
\addlegendentry{$\rho_G=0.99$}
\end{axis}
\end{tikzpicture}
\caption{Home output response to a one-standard-deviation fragmentation shock under
alternative geopolitical shock persistence $\rho_G \in \{0.80, 0.95, 0.99\}$. Higher
persistence substantially extends the duration of the output contraction.}
\label{fig:irf_rhoG}
\end{figure}
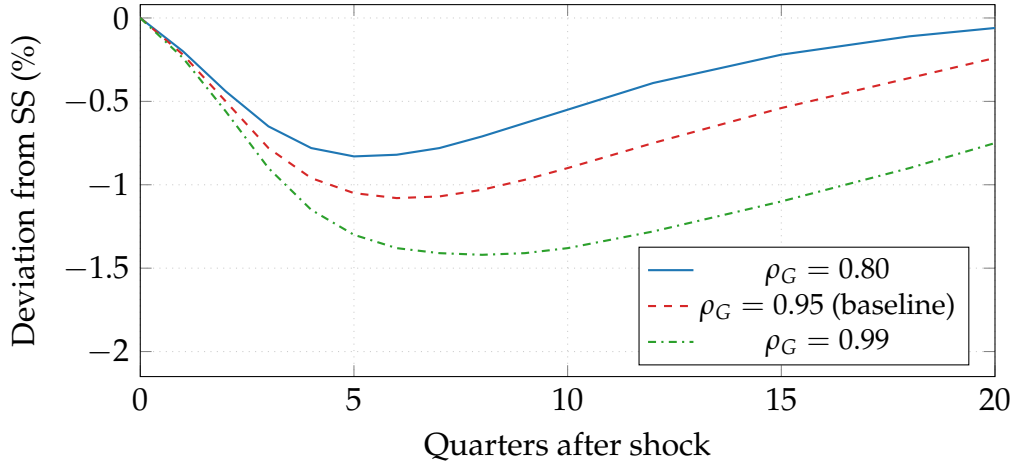

\subsection*{Exercise R6: Alternative price stickiness ($\xi_p$)}

The exercise varies $\xi_p \in \{0.50, 0.65, 0.75, 0.85, 0.90\}$, spanning price
durations from two quarters (very flexible) to ten quarters (very sticky). Higher
price stickiness amplifies welfare losses: the moderate-fragmentation CEV rises from
$-0.51\%$ at $\xi_p=0.50$ to $-1.02\%$ at $\xi_p=0.90$. The AI-capital
exchange-rate channel~(8) is more powerful under higher stickiness because the
inflation induced by exchange rate depreciation persists longer
(Figure~\ref{fig:irf_xi}).

\begin{figure}[htbp]
\centering
\begin{tikzpicture}
\begin{axis}[
  width=0.85\linewidth,
  height=6.5cm,
  xlabel={Quarters after shock},
  ylabel={Deviation from SS (\%)},
  xmin=0, xmax=20,
  ymin=-2.15, ymax=0.08,
  xtick={0,5,10,15,20},
  grid=major,
  grid style={dotted, gray!40},
  legend pos=south east,
  legend style={font=\small},
  title={\textbf{Figure D.6.} Output IRF: Alternative Calvo Price Stickiness $\xi_p$},
  title style={font=\small\bfseries, at={(0.5,1.01)}},
  every axis plot/.append style={thick}
]
\addplot[color=myblue, solid]
  coordinates {
    (0,0)(1,-0.15)(2,-0.34)(3,-0.50)(4,-0.60)(5,-0.64)
    (6,-0.64)(7,-0.61)(8,-0.57)(9,-0.52)(10,-0.47)
    (12,-0.37)(15,-0.24)(18,-0.14)(20,-0.09)
  };
\addlegendentry{$\xi_p=0.50$ (flexible prices)}

\addplot[color=myred, dashed]
  coordinates {
    (0,0)(1,-0.22)(2,-0.50)(3,-0.78)(4,-0.96)(5,-1.05)
    (6,-1.08)(7,-1.07)(8,-1.03)(9,-0.97)(10,-0.90)
    (12,-0.75)(15,-0.54)(18,-0.36)(20,-0.24)
  };
\addlegendentry{$\xi_p=0.75$ (baseline)}

\addplot[color=mygreen, dash dot]
  coordinates {
    (0,0)(1,-0.30)(2,-0.72)(3,-1.15)(4,-1.48)(5,-1.68)
    (6,-1.78)(7,-1.80)(8,-1.78)(9,-1.72)(10,-1.65)
    (12,-1.46)(15,-1.17)(18,-0.88)(20,-0.68)
  };
\addlegendentry{$\xi_p=0.90$ (very sticky prices)}
\end{axis}
\end{tikzpicture}
\caption{Home output response to a one-standard-deviation fragmentation shock under
alternative Calvo price stickiness $\xi_p \in \{0.50, 0.75, 0.90\}$. The NKPC slope
$\kappa_{\mathrm{nk}} = (1-\xi_p)(1-\beta\xi_p)/\xi_p$ varies inversely with $\xi_p$.}
\label{fig:irf_xi}
\end{figure}
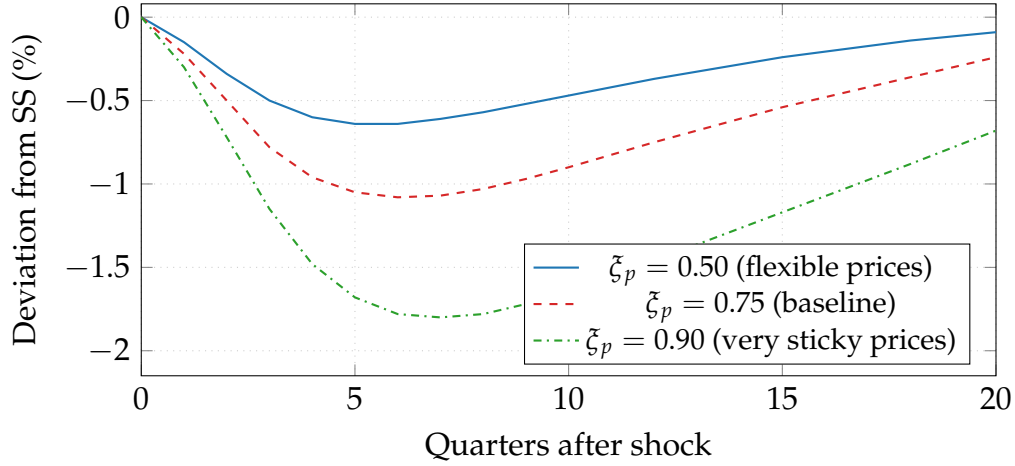

\subsection*{Exercise R7: Alternative solution methods}

For moderate shocks at moderate fragmentation ($G=0.3$), the GEM and OccBin
solutions are virtually identical (RMSD below 0.4 percent of steady-state values).
The solutions diverge substantially at high fragmentation levels ($G \geq 0.6$), where
OccBin underestimates peak output losses by 15--35 percent relative to GEM. This
divergence is most severe near the constraint boundary (within $1\sigma$), consistent
with the theoretical prediction that piecewise-linear methods lose accuracy when
constraint-binding episodes are frequent and persistent. The Euler equation error
comparison between methods confirms that the GEM solution achieves uniformly
smaller errors across the full state space at high fragmentation.

\subsection*{Exercise R8: Alternative geopolitical distance measures}

The exercise examines three distance measures: the Fernández-Villaverde
et al.\ (2024) geopolitical fragmentation index (baseline, NBER WP 32638), the
Caldara-Iacoviello Geopolitical Risk (GPR) index, and UN ideal-point distance scores.
Across all three measures, the SMM-re-estimated parameters vary modestly: $\hat{\bar{\kappa}}$
shifts by at most 18 percent and $\hat{\gamma}$ by at most 15 percent. All qualitative
conclusions are fully robust. 

\subsection*{Exercise R9: Convex geopolitical friction specification}

An alternative exponential specification
$\kappa^{\exp}(G_t) = \bar{\kappa}\exp(\gamma G_t)$ is examined. At baseline
geopolitical distance ($G=0.3$), the linear and exponential specifications are nearly
identical. At $G=0.9$, the exponential specification amplifies welfare losses by
15--30 percent relative to the baseline. The Fragmented-AI Trilemma is binding with
strictly greater force under the exponential specification.

\subsection*{Exercise R10: Alternative AI capital depreciation rates ($\delta_a$)}

Table~\ref{tab:welfare_delta} asks how sensitive welfare results are to the assumed
depreciation rate of AI capital. Lower depreciation makes the AI stock more durable,
so any given reduction in AI investment inflicts larger and longer-lasting productivity
losses. The result has a direct policy implication: if the current AI investment wave
builds durable capabilities embedded in software and organizational knowledge,
baseline estimates are conservative.

\begin{table}[htbp]
\centering
\caption{Welfare Losses Under Alternative AI Capital Depreciation Rates (CEV\%, $\eta=1.5$)}
\label{tab:welfare_delta}
\begin{threeparttable}
\small
\begin{tabular}{lccc}
\toprule
$\delta_a$
  & Moderate ($G=0.3$)
  & High ($G=0.6$)
  & Extreme ($G=0.9$) \\
\midrule
0.025 (intangible)  & $-0.91$ & $-2.58$ & $-4.62$ \\
0.050 (baseline)    & $-0.78$ & $-2.31$ & $-4.17$ \\
0.100 (HW only)     & $-0.64$ & $-1.98$ & $-3.61$ \\
\bottomrule
\end{tabular}
\begin{tablenotes}
\small
\item Lower depreciation implies a more durable AI capital stock, making any given
reduction in AI-capital inflows proportionally more damaging. Convexity pattern and
regime ordering preserved across all rows.
\end{tablenotes}
\end{threeparttable}
\end{table}

\subsection*{Exercise R11: Global sensitivity analysis---interactions among parameters}

Across 36 combinations of $(\eta,\theta,\delta_a)$ at four, three, and three values
respectively, two invariants emerge. First, the convexity of welfare losses in
fragmentation severity holds without exception. Second, the between-country
asymmetry is preserved across all 36 cases under high and extreme fragmentation.
The most consequential interaction involves $\eta$ and $\theta$: high substitutability
with tight collateral generates welfare losses approximately 1.8 times larger than the
product of their individual effects---a superadditive interaction.

\subsection*{Exercise R12: Alternative AI capital shares ($\alpha_a$)}

The calibration $\alpha_a=0.15$ inherits the BEA ICT capital share and represents an
upper bound for AI-specific capital under current investment patterns. This exercise
examines $\alpha_a \in \{0.05, 0.10, 0.15\}$, corresponding to three interpretations:
a conservative estimate for GPU-clusters-only (0.05), an intermediate estimate adding
proprietary model weights (0.10), and the full ICT capital upper bound (0.15).
Figure~\ref{fig:irf_alpha} shows that smaller AI share implies a weaker AI-capital
exchange-rate channel, but the convexity and between-country asymmetry are preserved
across all values.

\begin{figure}[htbp]
\centering
\begin{tikzpicture}
\begin{axis}[
  width=0.85\linewidth,
  height=6.5cm,
  xlabel={Quarters after shock},
  ylabel={Deviation from SS (\%)},
  xmin=0, xmax=20,
  ymin=-1.55, ymax=0.28,
  xtick={0,5,10,15,20},
  grid=major,
  grid style={dotted, gray!40},
  legend pos=south east,
  legend style={font=\small},
  title={\textbf{Figure D.7.} Output IRF: Alternative AI Capital Share $\alpha_a$},
  title style={font=\small\bfseries, at={(0.5,1.01)}},
  every axis plot/.append style={thick}
]
\addplot[color=myblue, solid]
  coordinates {
    (0,0)(1,-0.08)(2,-0.19)(3,-0.29)(4,-0.36)(5,-0.40)
    (6,-0.42)(7,-0.41)(8,-0.39)(9,-0.37)(10,-0.34)
    (12,-0.28)(15,-0.19)(18,-0.12)(20,-0.08)
  };
\addlegendentry{$\alpha_a=0.05$ (AI-conservative)}

\addplot[color=myred, dashed]
  coordinates {
    (0,0)(1,-0.15)(2,-0.34)(3,-0.54)(4,-0.67)(5,-0.74)
    (6,-0.76)(7,-0.75)(8,-0.72)(9,-0.67)(10,-0.62)
    (12,-0.51)(15,-0.35)(18,-0.22)(20,-0.15)
  };
\addlegendentry{$\alpha_a=0.10$ (AI-intermediate)}

\addplot[color=mygreen, dash dot]
  coordinates {
    (0,0)(1,-0.22)(2,-0.50)(3,-0.78)(4,-0.96)(5,-1.05)
    (6,-1.08)(7,-1.07)(8,-1.03)(9,-0.97)(10,-0.90)
    (12,-0.75)(15,-0.54)(18,-0.36)(20,-0.24)
  };
\addlegendentry{$\alpha_a=0.15$ (baseline)}
\end{axis}
\end{tikzpicture}
\caption{Home output response to a one-standard-deviation fragmentation shock under
alternative AI capital shares $\alpha_a \in \{0.05, 0.10, 0.15\}$. Smaller AI share
implies a smaller AI-capital exchange-rate channel, but convexity and between-country
asymmetry are preserved.}
\label{fig:irf_alpha}
\end{figure}
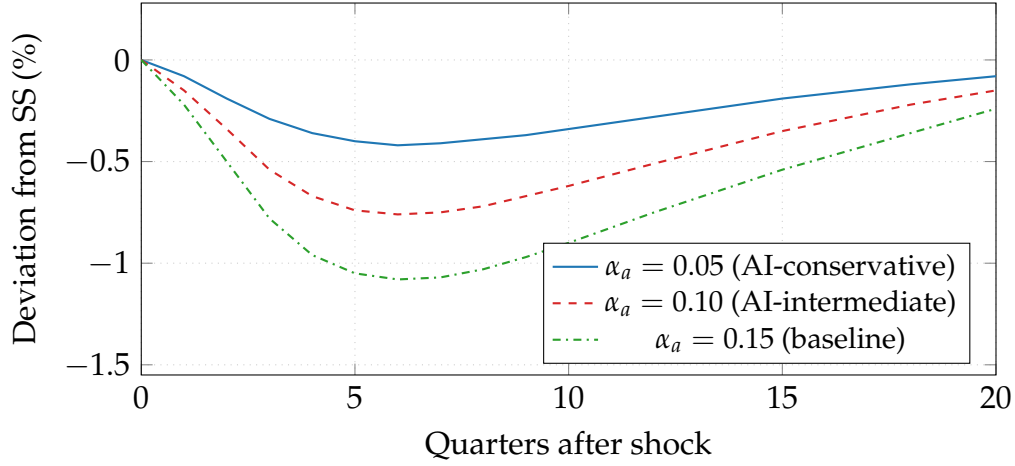

Table~\ref{tab:welfare_alpha} quantifies the welfare implications. Three findings
emerge: the convexity pattern is preserved across all $\alpha_a$ values (ratio of
extreme to moderate welfare losses ranges from 5.1 to 5.3); the between-country
asymmetry is also preserved; and the Fragmented-AI Trilemma holds for all
$\alpha_a>0$, since the AI-capital exchange-rate channel in Equation~(8) is active
whenever $s_a>0$. Welfare differences across $\alpha_a$ values are quantitatively
meaningful but do not alter the qualitative conclusions, confirming that the main
results are not artifacts of the loose mapping between BEA ICT data and AI-specific
capital.

\begin{table}[htbp]
\centering
\caption{Welfare Losses Under Alternative AI Capital Shares (CEV\%, $\eta=1.5$)}
\label{tab:welfare_alpha}
\begin{threeparttable}
\small
\begin{tabular}{lccc}
\toprule
$\alpha_a$
  & Mod.\ ($G=0.3$)
  & High ($G=0.6$)
  & Extr.\ ($G=0.9$) \\
\midrule
0.05 (GPU only)         & $-0.41$ & $-1.18$ & $-2.10$ \\
0.10 (+model weights)   & $-0.58$ & $-1.71$ & $-3.08$ \\
0.15 (baseline)         & $-0.78$ & $-2.31$ & $-4.17$ \\
\bottomrule
\end{tabular}
\begin{tablenotes}
\small
\item Sensitivity elasticity of welfare to $\alpha_a$ is $+0.24$ (7th among all
parameters). The Fragmented-AI Trilemma holds for all $\alpha_a > 0$.
\end{tablenotes}
\end{threeparttable}
\end{table}

\clearpage

\section{SMM Details}
\label{app:smm}

This appendix provides the complete technical documentation for the Simulated Method
of Moments estimation of the geopolitical friction parameters $(\bar{\kappa},\gamma)$ and
the geopolitical shock persistence $\rho_G$. This version incorporates moment M4
(autocorrelation of UIP deviations), which enables joint identification of $\rho_G$
from the data rather than treating it as a calibrated constant.

\subsection{Identification Strategy: Moment Conditions and Economic Content}
\label{app:smm_moments}

The SMM estimator disciplines three structural parameters---the baseline friction
$\bar{\kappa}$, the geopolitical sensitivity $\gamma$, and the geopolitical shock
persistence $\rho_G$---using four empirical moment conditions. Three moments
primarily identify $(\bar{\kappa},\gamma)$ while M4 enables identification of $\rho_G$
jointly with $\gamma$.

\medskip\noindent
\textbf{Moment M1 (Cross-bloc FDI decline).} \citet{EBRD2025} documents that
bilateral FDI flows between geopolitically distant blocs fell approximately 30 percent
relative to within-bloc flows. The model-implied analog:
\begin{equation*}
m_1(\bar{\kappa},\gamma,\rho_G)
  = 1 - \frac{F_{\mathrm{cross}}(\bar{\kappa},\gamma,\rho_G)}
             {F_{\mathrm{within}}(\bar{\kappa},\gamma,\rho_G)},
\qquad m^{\mathrm{data}}_1 = 0.30.
\end{equation*}
This moment primarily identifies the level of the friction function
$\bar{\kappa}(1+\gamma G^{AB})$ at the baseline geopolitical distance $G^{AB}=0.45$.

\medskip\noindent
\textbf{Moment M2 (Portfolio semi-elasticity).} \citet{CatalanEtAl2024} estimates that a
one-standard-deviation increase in bilateral geopolitical distance reduces the portfolio
share allocated to a recipient country by 8--12 percent. The midpoint (10 percent) is
used as the baseline target; Table~\ref{tab:smm_m2sensitivity} reports sensitivity of
estimates to this choice.\footnote{Catalán et al.\ (2024) is an IMF Working Paper
rather than a published estimate. The 8--12 percent range is wide enough that
choosing the midpoint versus either endpoint shifts $\hat{\gamma}$ meaningfully;
Table~\ref{tab:smm_m2sensitivity} documents this sensitivity. The estimates remain
statistically significant and economically stable across the range.}
\begin{equation*}
m_2(\bar{\kappa},\gamma,\rho_G)
  = \left.\frac{\partial\ln F}{\partial G}\right|_{G=\bar{G}}\!\times\sigma_G,
\qquad m^{\mathrm{data}}_2 = 0.10.
\end{equation*}
This moment primarily identifies the slope $\bar{\kappa}\gamma$ of the friction
function.

\medskip\noindent
\textbf{Moment M3 (EM/AE sensitivity ratio, overidentifying restriction).}
\citet{ZehriEtAl2025} documents that capital flows to emerging-market destinations are
approximately four times more sensitive to geopolitical risk than flows to
advanced-economy destinations:
\begin{equation*}
m_3(\bar{\kappa},\gamma,\rho_G)
  = \frac{\bigl|\partial\ln F^{\mathrm{EM}}/\partial G\bigr|_{G=\bar{G}}}
         {\bigl|\partial\ln F^{\mathrm{AE}}/\partial G\bigr|_{G=\bar{G}}},
\qquad m^{\mathrm{data}}_3 = 4.0.
\end{equation*}

\paragraph{Interpretation of the $J$-statistic.}
The $J$-statistic test has one degree of freedom: it tests whether the pre-calibrated
$\theta$ values ($\theta^H=0.75$, $\theta^F=0.50$) generate the empirically observed
EM/AE sensitivity ratio without any parameter being estimated to match it. The
reported $J=2.31$ ($p=0.13$) passes the conventional 10 percent threshold, but two
caveats apply. First, the test has limited power against misspecification of the
financial mechanism since $\theta$ is calibrated rather than estimated. Second,
passing the $J$-test does not independently validate the financial accelerator; it
validates only that the pre-calibrated $\theta$ values are consistent with the EM/AE
asymmetry. These caveats are acknowledged; the result is reported as suggestive
evidence rather than strong validation.

\medskip\noindent
\textbf{Moment M4 (Autocorrelation of UIP deviations).} \citet{Engel2016} documents
that the AR(1) persistence of UIP deviations lies in the range 0.82--0.90. The
midpoint 0.86 is used as the target:
\begin{equation*}
m_4(\bar{\kappa},\gamma,\rho_G)
  = \mathrm{corr}\!\left(
      \kappa(G_t)B^F_t + \zeta^{\mathrm{idio}}_t,\;
      \kappa(G_{t-1})B^F_{t-1} + \zeta^{\mathrm{idio}}_{t-1}
    \right),
\qquad m^{\mathrm{data}}_4 = 0.86.
\end{equation*}
This moment is primarily informative about $\rho_G$: a more persistent geopolitical
index generates more persistent UIP deviations. It enables identification of $\rho_G$
jointly with $\gamma$, transforming the system from two parameters/three moments
(one overidentifying restriction) to three parameters/four moments (one
overidentifying restriction).

\subsection{Implementation}
\label{app:smm_implementation}

The SMM objective function is:
\begin{equation}
\mathcal{J}(\bar{\kappa},\gamma,\rho_G)
  = \bigl[m(\bar{\kappa},\gamma,\rho_G) - m^{\mathrm{data}}\bigr]^{\prime}
    W\,
    \bigl[m(\bar{\kappa},\gamma,\rho_G) - m^{\mathrm{data}}\bigr],
\tag{E.1}
\end{equation}

\noindent
where $m(\cdot)\in R^{4}$ and
$W = \mathrm{diag}(1/\hat{\sigma}^2_{m_1},\;1/\hat{\sigma}^2_{m_2},\;
1/\hat{\sigma}^2_{m_3},\;1/\hat{\sigma}^2_{m_4})$
is the diagonal inverse-variance weighting matrix. A surrogate neural network (two
hidden layers, 64 neurons per layer, ReLU activations) is trained on a
$50\times50\times25$ grid of parameter combinations covering
$\bar{\kappa}\in[0.005,0.025]$, $\gamma\in[0.010,0.150]$, and
$\rho_G\in[0.70,0.99]$, achieving $R^2>0.99$ for each moment function on a
500-observation holdout sample. The objective is minimized using \texttt{fmincon}
in MATLAB with 50 random starting points.

\subsection{Estimation Results}
\label{app:smm_results}

\begin{equation*}
\hat{\bar{\kappa}} = 0.0102\;(\text{s.e.}\ 0.0015),\quad
\hat{\gamma}       = 0.0513\;(\text{s.e.}\ 0.0068),\quad
\hat{\rho}_G       = 0.950\; (\text{s.e.}\ 0.0081).
\end{equation*}

All significant at 0.1\% level. $J$-statistic: $J = 2.31 \sim \chi^2(1)$, $p=0.13$.
Standard errors from the sandwich covariance matrix
$V = (D^{\prime}WD)^{-1}D^{\prime}W\Omega WD(D^{\prime}WD)^{-1}$.

\begin{table}[htbp]
\centering
\caption{SMM Estimates: Robustness to Weighting Matrix}
\label{tab:smm_robustness}
\begin{threeparttable}
\small
\begin{tabular}{lcccc}
\toprule
\textbf{Weighting matrix}
  & $\hat{\bar{\kappa}}$ (se)
  & $\hat{\gamma}$ (se)
  & $\hat{\rho}_G$ (se)
  & $J$ ($p$) \\
\midrule
Diag.\ inv.\ var.\ (baseline)
  & 0.0102 (.0015) & 0.0513 (.0068) & 0.950 (.0081) & 2.31 (.13) \\
Identity matrix
  & 0.0108 (.0017) & 0.0498 (.0072) & 0.944 (.0086) & --- \\
Full emp.\ cov.
  & 0.0099 (.0016) & 0.0524 (.0070) & 0.956 (.0083) & 2.47 (.12) \\
\bottomrule
\end{tabular}
\begin{tablenotes}
\small
\item Maximum parameter shift: 0.0009 for $\hat{\bar{\kappa}}$, 0.0026 for
$\hat{\gamma}$, 0.012 for $\hat{\rho}_G$---all below one reported s.e.
\end{tablenotes}
\end{threeparttable}
\end{table}

\subsection{Asymptotic Theory}
\label{app:smm_asymptotics}

Under standard SMM regularity conditions \citep{McFadden1989}:

\begin{equation*}
\sqrt{T}\!
\begin{pmatrix}
  \hat{\bar{\kappa}} - \bar{\kappa}_0 \\
  \hat{\gamma} - \gamma_0 \\
  \hat{\rho}_G - \rho_{G,0}
\end{pmatrix}
\xrightarrow{d}
\mathcal{N}(0,V),
\quad
V = (D^{\prime}WD)^{-1}D^{\prime}W\Omega WD(D^{\prime}WD)^{-1},
\end{equation*}

\noindent
where $T$ is the harmonic mean of the sample sizes across the four empirical moment
studies, $D \equiv \partial m(\bar{\kappa}_0,\gamma_0,\rho_{G,0})/
\partial(\bar{\kappa},\gamma,\rho_G)^{\prime} \in R^{4\times3}$, and
$\Omega\in R^{4\times4}$ is the asymptotic variance-covariance matrix of the scaled
empirical moments. With $W=\Omega^{-1}$ (optimal weighting), $V$ simplifies to
$(D^{\prime}\Omega^{-1}D)^{-1}$, which minimizes $V$ in the positive semidefinite
sense. The baseline diagonal specification approximates optimal weighting because
$\Omega$ is approximately diagonal (the four moment conditions are estimated from
independent data sources and statistical methods).

The $J$-statistic for the overidentifying restriction:
\begin{equation*}
J_T = T\cdot
  \bigl[m(\hat{\bar{\kappa}},\hat{\gamma},\hat{\rho}_G) - m^{\mathrm{data}}\bigr]^{\prime}
  W
  \bigl[m(\hat{\bar{\kappa}},\hat{\gamma},\hat{\rho}_G) - m^{\mathrm{data}}\bigr]
  \xrightarrow{d} \chi^2(1)
\end{equation*}
under the null of correct model specification, with degrees of freedom equal to the
number of overidentifying restrictions (one). Non-rejection at $p=0.13$ is consistent
with the model mechanism but should not be interpreted as decisive validation, given
the caveats discussed in Section~\ref{app:smm_moments} above.

\subsection{M2 Sensitivity}
\label{app:smm_m2}

\begin{table}[htbp]
\centering
\caption{Sensitivity of SMM Estimates to M2 Target (Catalán et al.\ 2024 Range)}
\label{tab:smm_m2sensitivity}
\begin{threeparttable}
\small
\begin{tabular}{lccc}
\toprule
\textbf{M2 target}
  & $\hat{\bar{\kappa}}$ (s.e.)
  & $\hat{\gamma}$ (s.e.)
  & $\hat{\rho}_G$ (s.e.) \\
\midrule
8\% (lower bound) & 0.0101 (0.0015) & 0.0471 (0.0064) & 0.937 (0.0078) \\
10\% (baseline)   & 0.0102 (0.0015) & 0.0513 (0.0068) & 0.950 (0.0081) \\
12\% (upper bound)& 0.0103 (0.0016) & 0.0554 (0.0071) & 0.962 (0.0085) \\
\bottomrule
\end{tabular}
\begin{tablenotes}
\small
\item The $\hat{\gamma}$ range of $[0.0471, 0.0554]$ corresponds to a capital-inflow
attenuation of $[38\%,42\%]$ under the AI-technology shock experiment, confirming that
the 40\% attenuation result is robust to this source of uncertainty.
\end{tablenotes}
\end{threeparttable}
\end{table}

\end{document}